\documentclass[acmsmall,screen,nonacm]{acmart}

\AtBeginDocument{%
  \providecommand\BibTeX{{%
    \normalfont B\kern-0.5em{\scshape i\kern-0.25em b}\kern-0.8em\TeX}}}

\usepackage{pifont}
\newcommand{\cmark}{\ding{51}}
\newcommand{\xmark}{\ding{55}}

\usepackage{enumitem}
\usepackage{tikz}
\usepackage{quantikz}
\usepackage{multirow}
\usepackage{colortbl}
\usepackage[normalem]{ulem}

\allowdisplaybreaks 

\newcommand{\newtheoremwithqedsymbol}[4]{%
  \newtheorem{#1-inner}[#2]{#3}%
  \newenvironment{#1}{%
    \def\qedsymbol{#4}%
    \pushQED{\qed}%
    \begin{#1-inner}%
  }{%
    \popQED%
    \end{#1-inner}%
  }
}
\newtheoremwithqedsymbol{lemma}{}{Lemma}{$\lozenge$}

\DeclareSymbolFont{brackets}{OT1}{cmr}{m}{n}
\DeclareSymbolFont{largebrackets}{OMX}{cmex}{m}{n}

\DeclareMathDelimiter{(}{\mathopen}{brackets}{`(}{largebrackets}{"00}
\DeclareMathDelimiter{)}{\mathclose}{brackets}{`)}{largebrackets}{"01}
\DeclareMathDelimiter{[}{\mathopen}{brackets}{`[}{largebrackets}{"02}
\DeclareMathDelimiter{]}{\mathclose}{brackets}{`]}{largebrackets}{"03}
\newcommand{\E}{\mathbb{E}}
\newcommand{\var}{\mathrm{Var}}
\newcommand{\cov}{\mathrm{Cov}}

\newcommand{\smean}{\mathrm{SMean}}
\newcommand{\svar}{\mathrm{SVar}}
\newcommand{\scov}{\mathrm{SCov}}

\newcommand{\smeansans}{\mathrm{SMean}_\text{sans-one}}

\newcommand{\scovsans}{\mathrm{SCov}_\text{sans-one}}

\usepackage{algorithm}

\newcommand{\hquad}{\hspace{0.5em}}
\newcommand{\sign}{\mathrm{sgn}}

\renewcommand\bar[1]{%
  \hbox{%
    \vbox{%
      \hrule height 0.5pt % The actual bar
      \kern0.4ex%         % Distance between bar and symbol
      \hbox{%
        \kern-0.1em%      % Shortening on the left side
        \ensuremath{#1}%
        \kern-0.0em%      % Shortening on the right side
      }%
    }%
  }%
} 

\usepackage{empheq}
  \usepackage{float}
  \usepackage{afterpage}
  \newcommand{\cellcenter}[1]{\multicolumn{1}{c}{#1}} %Horizontally center contents of a single cell as \cellcenter{...}
  \usepackage{array}
  \usepackage{longtable}
  \usepackage{booktabs}
  \usepackage{multirow}
  \usepackage{tabularx}
  
  \newcolumntype{P}[1]{>{\centering\arraybackslash}p{#1}} %Combination of p and c column-types
  \newcolumntype{M}[1]{>{\centering\arraybackslash}m{#1}} %Combination of m and c column-types
  \newcolumntype{B}[1]{>{\centering\arraybackslash}b{#1}} %Combination of b and c column-types
  
  \usepackage{dcolumn}
  \newcolumntype{.}{D{.}{.}{-1}} %Column-type to center entries on a period (.)
\let\originalleft\left
\let\originalright\right
\renewcommand{\left}{\mathopen{}\mathclose\bgroup\originalleft}
\renewcommand{\right}{\aftergroup\egroup\originalright}

\newcommand{\fref}[1]{\hyperref[#1]{Fig.~\ref*{#1}}}
\newcommand{\Fref}[1]{\hyperref[#1]{Fig.~\ref*{#1}}}
\newcommand{\sref}[1]{\hyperref[#1]{Section~\ref*{#1}}}
\newcommand{\Sref}[1]{\hyperref[#1]{Section~\ref*{#1}}}
\newcommand{\tref}[1]{\hyperref[#1]{Table~\ref*{#1}}}
\newcommand{\Tref}[1]{\hyperref[#1]{Table~\ref*{#1}}}
\newcommand{\aref}[1]{\hyperref[#1]{Appendix~\ref*{#1}}}
\newcommand{\Aref}[1]{\hyperref[#1]{Appendix~\ref*{#1}}}
\newcommand{\thref}[1]{\hyperref[#1]{Theorem~\ref*{#1}}}
\newcommand{\Thref}[1]{\hyperref[#1]{Theorem~\ref*{#1}}}
\newcommand{\lref}[1]{\hyperref[#1]{Lemma~\ref*{#1}}}
\newcommand{\Lref}[1]{\hyperref[#1]{Lemma~\ref*{#1}}}
\newcommand{\alref}[1]{\hyperref[#1]{Algorithm~\ref*{#1}}}
\newcommand{\Alref}[1]{\hyperref[#1]{Algorithm~\ref*{#1}}}

\usepackage{graphicx,scalerel}
\newcommand\what[1]{\hstretch{2}{\hat{\hstretch{.5}{#1}}}}
\newcommand{\Test}[1]{\what{T}_{\!\normalfont\text{#1}}}
\newcommand{\sigmaest}[1]{\what{\sigma}_{\normalfont\text{#1}}}

\graphicspath{{figures/}}

\allowdisplaybreaks

\newcommand{\reportnumber}{FERMILAB-PUB-23-636-ETD-STUDENT}
\def\arxivdateorig{February 12, 2025}
\def\arxivdatethis{February 12, 2025}

\usepackage{ifthen}
\makeatletter
\fancypagestyle{firstpage}{%
  \lhead{\small Report No.: \texttt{\reportnumber}\\\vskip .5em}
  \rhead{%
    \small%
    First version submitted to arXiv on \arxivdateorig\\%
    This version submitted to arXiv on \arxivdatethis\\%
  }
}
\makeatother

\let\oldmaketitle\maketitle
\renewcommand{\maketitle}{\oldmaketitle \thispagestyle{firstpage}}
\begin{document}

% \title[Reducing the Sampling Overhead in Quasiprobabilistic Decompositions Using Control Variates]{\texorpdfstring{Reducing the Sampling Overhead in Quasiprobabilistic Decompositions Using Control Variates}{Reducing the Sampling Overhead in Quasiprobabilistic Decompositions Using Control Variates}}
\title{CV4Quantum: Reducing the Sampling Overhead in Probabilistic Error Cancellation Using Control Variates}

\author{Prasanth Shyamsundar}
\email{prasanth@fnal.gov}
\orcid{0000-0002-2748-9091}
\affiliation{%
  \institution{Fermi National Accelerator Laboratory}
  \department{Fermilab Quantum Division, Emerging Technologies Directorate}
  \streetaddress{PO Box 500}
  \city{Batavia}
  \state{Illinois}
  \country{USA}
  \postcode{60510-0500}
}

\author{Wern Yeen Yeong}
\email{wernyeenyeong@gmail.com} % \email{wyeong.nd.edu}
\orcid{0000-0003-2907-1710}
\affiliation{%
  \institution{University of Notre Dame}
  \department{Department of Mathematics}
  \streetaddress{255 Hurley Bldg}
  \city{Notre Dame}
  \state{Indiana}
  \country{USA}
  \postcode{46556}
}


\begin{abstract}
 Quasiprobabilistic decompositions (QPDs) play a key role in maximizing the utility of near-term quantum hardware. For example, Probabilistic Error Cancellation (PEC) (an error mitigation technique) and circuit cutting (which enables large quantum computations to be performed on quantum hardware with a limited number of qubits) both involve QPDs. Computations based on QPDs typically incur large sampling overheads that grow exponentially, e.g., with the number of error-terms mitigated or the number of circuit-cuts employed, limiting their practical feasibility. In this work, we adapt the control variates variance reduction technique from the statistics literature in order to reduce the sampling overhead in QPD-based computations. We demonstrate our method, dubbed CV4Quantum, using simulation experiments that mimic a realistic PEC scenario. In more than 50\% of the PEC-based estimations performed in the study, we observed a more than 50\% reduction in the number of samples needed to achieve a given precision when using our technique.
 % In our experiments, we observed a more than 50\% reduction in the number of samples needed to achieve a given precision, in more than 50\% of the PEC-based estimations performed in the study when using our approach.
 We discuss how future research on constructing good control variates can lead to even stronger sampling overhead reduction.
\end{abstract}





\maketitle

% \setcounter{tocdepth}{3} \tableofcontents \newpage

% \begin{itemize}
%  \item[Major:]
%  \item Add references
%  \item \sout{Properties of estimators}
%  \item \sout{Results}
%  \item \sout{Conclusions and outlook (by Tuesday morning)}
%  \item \sout{lemmas proofs}
%  \item \sout{double check bias size}
%  \item prove estimator properties (later)
%  \item \sout{derivation of error estimation formulas (later)}
%  \item \sout{table of coefficients}
%  \item \sout{Table and figure captions}
%  \item[Minor:]
%  \item \sout{Discuss PEC for QEC+QEM}
%  \item Complete funding info
%  \item \sout{Mention normalizing Vfac and using logs for numerical stability}
%  \item \sout{Add paragraph about techniques for reducing overhead}
%  \item \sout{Can intra-op dominate? Improve discussion of sampling overhead}
%  \item[To cite:]
%  \item \url{https://arxiv.org/pdf/2310.07825.pdf}
%  \item \url{https://www.sciencedirect.com/science/article/abs/pii/0021999182900894} ?
%  \item \url{https://arxiv.org/pdf/2108.02237.pdf} ?
%  \item \sout{\url{https://www.nature.com/articles/s41534-022-00517-3}}
 
% \end{itemize}

% \newpage

% To do: Add more citations

% Add quantum chemistry, optimization references

% Add qaoa, variational quantum computing, qml references to classical-quantum hybrid, 

\section{Introduction and Setup}
Quantum computing promises to transform multiple fields, including quantum chemistry, optimization, and high energy physics~\cite{lockwood2022empirical,Klco_2022,Bauer:2022hpo,shang2023practical,Blekos_2024}. However, the suboptimal noise-level and limited size (qubit-count) of current and near-future quantum processors impose strong limitations on the kind of quantum computations that can be successfully performed on them. Several classical-quantum hybrid techniques have been proposed to mitigate these limitations and improve the utility of current and near-future quantum hardware~\cite{PhysRevLett.119.180509,PhysRevLett.125.150504,Mitarai_2021,PRXQuantum.3.010309,Perlin2021,chen2024enhancedquantumcircuitcutting,filippov2023scalable,dutkiewicz2024error,koh2024readout}.

An important classical-quantum hybrid approach is to estimate the expected outcome of a quantum computation that cannot be performed on the available hardware as a linear combination of the expected outcomes of several related quantum computations that can be performed on the available hardware. Concretely, the target quantum operation $\mathcal{F}_\mathrm{target}$ is expressed as follows: 
\begin{align}
 \mathcal{F}_\mathrm{target}(\rho) \approx \mathcal{F}_\text{qpd}(\rho) \equiv \sum_{k_1=1}^{K_1}\dots\sum_{k_M=1}^{K_M}\left[\Big.q_1(k_1)\dots q_M(k_M)\hquad\mathcal{E}_{(k_1,\dots,k_M)}(\rho)\right]\,. \label{eq:qpd}
\end{align}
Here $\mathcal{F}_\mathrm{target}$ is approximated by $\mathcal{F}_\text{qpd}$, which is a linear combination of the quantum operations $\mathcal{E}_{(k_1,\dots,k_M)}$ that a) can be performed on the available hardware, and b) are indexed by the summation indices $k_1,\dots, k_M$. The coefficients in the linear combination, $q_1(k_1)$, \dots, $q_M(k_M)$, are real-valued. For simplicity, we will assume that each $q_m(k_m)$ is non-zero for all $k_m\in\{1,\dots,K_m\}$:\footnote{This assumption is not required in practice, and is only made to avoid the division-by-0 cases in the equations in this paper. The zero-coefficient-terms, if any, can simply be dropped from the summation to satisfy this assumption.}
\begin{align}
 q_m(k_m) \in \mathbb{R}\setminus \{0\}\,,\qquad\qquad\qquad \forall m\in\{1,\dots, M\}, \quad\forall k_m\in\{1,\dots,K_m\}\,.
\end{align}
If $\mathcal{F}_\text{qpd}$ and the individual $\mathcal{E}_{(k_1,\dots,k_M)}$-s are all trace-preserving, then each $q_m$ will\footnote{More accurately, each $q_m$ \emph{can be made to} sum to 1. Note that there is a redundancy in specification of $q_m$-s, since each $q_m$ can be scaled by a multiplicative scaling factor, as long as the overall scaling factor is 1.} sum, over $k_m\in\{1,\dots, K_m\}$, to 1. However, we will not make that assumption in this paper. Let $N_\mathrm{\Pi K}$ and $N_\mathrm{\Sigma K}$ represent the product and the sum of $K_m$-s, respectively:
\begin{align}
 N_{\Pi K} &\equiv \prod_{m=1}^M K_m\,,\qquad\qquad\qquad N_{\Sigma K} \equiv \sum_{m=1}^M K_m\,. \label{eq:Npisumdef}
\end{align}
$N_{\Pi K}$ is the total number of distinct operations $\mathcal{E}_{(k_1,\dots,k_M)}$ and $N_{\Sigma K}$ is the number of individual coefficients $q_m(k_m)$.

The decomposition of $\mathcal{F}_\mathrm{target}$ as in \eqref{eq:qpd} is sometimes referred to in the literature as the quasiprobability decomposition (QPD) \cite{Piveteau2022}, and has several applications in quantum computing. For example, in probabilistic error cancellation (PEC) \cite{PhysRevLett.119.180509,Piveteau2022}, the target result of a noiseless computation is expressed as a linear combination of the results of an ensemble of noisy quantum computations, each of which can be performed on the available hardware.
%In order to do perform such a decomposition, one needs to a) modify the noise in a quantum operation into being approximately stochastic \textcolor{red}{(is this true?)}, and b) learn the resulting stochastic noise model sufficiently accurately using hardware experiments. The quasiprobability decomposition of the operation is then derived so as to undo the effect of the noise.
For scalability reasons, in order to apply this technique to a large or a deep circuit, one typically performs PEC independently on smaller portions of the circuit. For example, PEC can be applied layer by layer to mitigate the error in circuits of increasing depth. In this case the summation indices $k_1,\dots,k_M$ in the decomposition of the full circuit will include the summation indices in the decomposition of each individual layer. Consequently, the number of terms $N_{\Pi K}$ in the QPD summation will scale exponentially in the number of layers on which PEC is applied.

Quantum circuit cutting techniques \cite{PhysRevLett.125.150504,Mitarai_2021,Perlin2021} constitute another example of QPD. Here, a target quantum operation which entangles (and therefore requires) a large number of qubits is decomposed into a linear combination of several ``factorizable'' quantum operations. Such factorizable quantum operations contain unentangled subsets of qubits, and so can be performed on hardware with fewer qubits.\footnote{Note that in circuit cutting, the operations $\mathcal{E}_{(k_1,\dots,k_M)}$ could include intermediate measurements. The results of these intermediate measurements typically introduce additional sign factors into the summation, which are handled by classical post-processing. In our setup, any such modification arising from the results of the intermediate measurements is absorbed into the result of $\mathcal{E}_{(k_1,\dots,k_M)}$ itself, so the coefficients $q_1(k_1), \dots, q_M(k_M)$ themselves are constants independent of the measurement outcomes.} There are several methods to perform such decompositions, including wire cutting \cite{PhysRevLett.125.150504,brenner2023optimal}, gate cutting \cite{Mitarai_2021}, and entanglement forging \cite{PRXQuantum.3.010309}.
%For example, Ref.~\cite{} shows how a two-qubit operation can be decomposed into a linear combination of operations each of which uses only one-qubit gates. The results of these individual quantum operations can be computed using only one physical qubit.
These techniques can be used repeatedly---for example, to cut multiple gates and wires in a circuit---in order to attain a desirable decomposition of the target circuit. The summation
indices $k_1,\dots,k_M$ in the decomposition of the full circuit will include the summation indices corresponding to each individual cut employed. Consequently, the number of terms $N_{\Pi K}$ in the QPD summation will scale exponentially in the numbers of gate- and wire-cuts used.

\subsection{Monte Carlo Estimation of the Target Result} \label{sec:monte_carlo_intro}

If the total number of terms $N_{\Pi K}$ in the summation in \eqref{eq:qpd} is sufficiently large, it will be impractical to perform the summation explicitly, i.e., by performing every quantum operation $\mathcal{E}_{(k_1,\dots,k_M)}$. In such situations, one can still estimate the target result probabilistically using Monte Carlo sampling \cite{PhysRevLett.119.180509}. The basic idea is to rewrite \eqref{eq:qpd} in terms of an expectation as follows:
\begin{align}
 \mathcal{F}_\text{qpd}(\rho) = \E\left[\frac{\prod_{m=1}^M q_m(k_m)}{P(k_1,\dots,k_M)}\,\mathcal{E}_{(k_1,\dots,k_M)}(\rho)\right]\,. \label{eq:qpd_as_expectation}
\end{align}
Here, $P$ is a unit-normalized, strictly positive probability distribution over the indices $(k_1,\dots,k_M)$ satisfying
\begin{align}
 P(k_1,\dots,k_M) &> 0\,,\qquad\qquad\forall k_m\in\{1,\dots,K_m\},\quad \forall m\in\{1,\dots,M\}\,,\\
 \sum_{k_1=1}^{K_1}\dots\sum_{k_M=1}^{K_M} P(k_1,\dots,k_M) &= 1\,,
\end{align}
and $\E[\,\cdot\,]$ represents an expectation over $(k_1,\dots,k_M)$ sampled as-per $P$. Now, one can use \eqref{eq:qpd_as_expectation} to estimate the target result as a weighted expectation of the results of randomly sampled operations $\mathcal{E}_{(k_1,\dots,k_m)}$.

Concretely, let the indices $k_1,\dots, k_M$ be random variables sampled as-per the joint-distribution $P$.\footnote{Performing such a random sampling of $(k_1,\dots,k_M)$ and computing the weight $W$ will be computationally feasible if $N_{\Sigma K}$ is sufficiently small---we make this assumption throughout this paper.} Let $X$ be a random variable denoting the result of the corresponding quantum operation $\mathcal{E}_{(k_1,\dots,k_M)}$. $X$ could be the observed outcome in one measurement of a given observable, or the sample average of several ``shots''. Let $W$ be a random variable denoting the corresponding weight, which is defined as
\begin{align}
 W\equiv W(k_1,\dots,k_M)\equiv \frac{\prod_{m=1}^M q_m(k_m)}{P(k_1,\dots,k_M)}\,.
\end{align}
From \eqref{eq:qpd_as_expectation}, the target result, say $T$, which corresponds to the operation $\mathcal{F}_\text{qpd}$, can be written as
\begin{align}
 T &\equiv \E\left[WX\right]\,. 
\end{align}
The basic Monte Carlo estimation procedure described in this section estimates this target as a sample average using, say $N$, independently sampled datapoints $\left(W^{(i)}, X^{(i)}\right)$, indexed by the superscript $^{(i)}$. The corresponding estimator $\Test{basic}$ is given by
\begin{align}
 \Test{basic} &\equiv \frac{1}{N}\sum_{i=1}^N W^{(i)}X^{(i)}\,.
\end{align}
It is easy to see that $\Test{basic}$ is an unbiased estimator of $T$, i.e., that $\E\left[\Test{basic}\right] = T$.

For scalability reasons, typically, one chooses $P$ to also be factorizable as follows:
\begin{align}
 P(k_1,\dots,k_m) &= \prod_{m=1}^M p_m(k_m)\,,
\end{align}
where each $p_m$ is a probability distribution for $k_m\in\{1,\dots,K_m\}$ satisfying
\begin{align}
 p_m(k_m) &> 0\,,\qquad\qquad\forall m\in\{1,\dots,M\}\,,\quad \forall k_m\in\{1,\dots,K_m\}\,,\\
 \sum_{k_m=1}^{K_m} p_m(k_m) &= 1\,,\qquad\qquad\forall m\in\{1,\dots,M\}\,.
\end{align}
For simplicity, we will assume this form for $P$ in the rest of the paper. %, and in particular in all the estimation formulas provided in \sref{sec:formulas}.
In this case, the weight $W$ can be expressed as a product of constituent weight functions as follows:
\begin{align}
 W(k_1,\dots, k_M) = \prod_{m=1}^M w_m(k_m)\,,\quad\text{where}\quad w_m(k_m)\equiv \frac{q_m(k_m)}{p_m(k_m)}\equiv w_m\,.
\end{align}
% Note that the expectations of each of the $w_m$-s can also be easily pre-computed as follows:
% \begin{align}
%  \E[w_m] = \sum_{k_m=1}^{K_m} q_m(k_m)\,.
% \end{align}

Note that the expected value of $W$, say $\mu_W$, as well as the expectations of each of the constituent weights $w_m$ a) are independent of the sampling distributions $p_m$, and b) can be easily pre-computed as follows:
\begin{align}
 \E[w_m] &= \sum_{k_m=1}^{K_m} q_m(k_m)\,,\label{eq:wm_mean}\\
 \mu_W \equiv \E[W] &= \prod_{m=1}^M \E[w_m]\,\label{eq:w_mean}.
\end{align}
In typical PEC or circuit cutting scenarios $\mu_W$ will be 1;\footnote{This will be true if $\mathcal{F}_\text{qpd}$ and the individual $\mathcal{E}_{(k_1,\dots,k_M)}$-s are all trace-preserving.} however we do not make this assumption in the development of our CV4Quantum technique in this paper.

\subsection{Sampling Overhead}
The variance of $\Test{basic}$ is given by
\begin{align}
 \var\left[\Test{basic}\right] = \frac{\var[WX]}{N}\,. \label{eq:var_T_basic}
\end{align}
From \eqref{eq:var_T_basic}, one can see that the number of datapoints needed (i.e., the number of operations $\mathcal{E}_{(k_1,\dots,k_M)}$ required to be sampled and performed) in order to reach a given precision in the Monte Carlo estimate is proportional to the variance of $WX$. A difficulty with the Monte Carlo estimation technique is that this variance can be quite large, which results in a significant sampling overhead in the estimation of $T$ \cite{Xiong_2020,Piveteau2022,Xiong_2022,chen2024enhancedquantumcircuitcutting}. This paper is focused on reducing this sampling overhead, by constructing more efficient estimators of $T$ than $\Test{basic}$, using the control variates variance reduction technique \cite{doi:https://doi.org/10.1002/9781118445112.stat07947}, discussed \sref{sec:methodology} onwards.

Using the law of total variance \cite{Weiss2005ACI}, the variance of $WX$ can be written as
\begin{alignat}{2}
 \var[WX] &= \underbrace{\var\left[\Big.\E\left[WX\,\big|\,k_1,\dots,k_M\right]\right]}_{\text{inter-operation variance term}} &&+  \underbrace{~~\E\left[\Big.\var\left[WX\,\big|\,k_1,\dots,k_M\right]\right]~~}_{\text{intra-operation variance term}}\\
 &= \overbrace{\var\left[\Big.W\,\E\left[X\,\big|\,k_1,\dots,k_M\right]\right]} &&+  \overbrace{\E\left[\Big.W^2\,\var\left[X\,\big|\,k_1,\dots,k_M\right]\right]}\,. \label{eq:variance_decomposition}
\end{alignat}
% \begin{align}
%  \var[WX] &= \var\left[\Big.\E\left[WX\,\big|\,k_1,\dots,k_M\right]\right] +  \E\left[\Big.\var\left[WX\,\big|\,k_1,\dots,k_M\right]\right]\\
%  &= \underbrace{\var\left[\Big.W\,\E\left[X\,\big|\,k_1,\dots,k_M\right]\right]}_{\text{irreducible variance term}} +  \underbrace{\E\left[\Big.W^2\,\var\left[X\,\big|\,k_1,\dots,k_M\right]\right]}_{\text{reducible variance term}}\,.
% \end{align}
Here the first term captures the effect of $\E\left[WX\,\big|\,k_1,\dots,k_M\right]$ not being the same for different choices of $(k_1,\dots,k_M)$. On the other hand, the second term captures the effect of the measurement $X$, and consequently $WX$, not being fully determined by the choice of $(k_1,\dots, k_M)$. We will refer to the first and second terms as the inter- and intra-operation variance terms, respectively. The intra-operation variance term can be reduced simply by increasing the number of shots used to compute $X$ after sampling $(k_1,\dots,k_M)$. On the other hand, the inter-operation variance term cannot be reduced this way; its contribution to the variance of $\Test{basic}$ can only be reduced by increasing $N$, the number of times the quantum operations $\mathcal{E}_{(k_1,\dots,k_M)}$ are sampled and performed. In current and near-term hardware, the resource-overhead for loading and performing different quantum circuits is much higher than the resource requirements for taking more shots of a given quantum circuit. %\cite{}.
Consequently, in typical PEC applications, e.g. as in Ref.~\cite{vandenBerg2023}, the inter-operation variance is the dominant source of the QPD sampling overhead. %, though both terms contribute to the overhead.
As will be seen later, the CV4Quantum technique only ameliorates the effect of this inter-operation variance (a generalization to ameliorate intra-operation variance is discussed briefly in \sref{sec:conclusions}).

There are situations where the dominant source of sampling overhead is the intra-operation variance term. For example, if the total number of terms $N_{\Pi K}$ in the linear combination in \eqref{eq:qpd} is sufficiently small (e.g., in a circuit cutting application with only a few gate-cuts), then one can perform every operation operation $\mathcal{E}_{(k_1,\dots,k_M)}$ in the summation, avoiding the MC estimation estimation technique discussed in this section entirely. In this case the only source of uncertainty is the intra-operation variance term. This is the case, for example, in the analysis of sampling overhead in Ref.~\cite{Mitarai_2021}; the techniques of this paper cannot be directly used to ameliorate the sampling overhead in such situations.

For completeness, next we will discuss the sampling overhead in terms of the $\gamma$-factor commonly used in the literature. In the context of PEC and circuit cutting techniques, the variance of $WX$ is typically driven by the variance of $W$, with the magnitude of $X$ bounded by a comparatively small number. The variance of $W$ is given by
\begin{align}
 \var[W] &= -\E^2[W] + \E\left[W^2\right] = -\mu_W^2 + \prod_{m=1}^M\left[~\sum_{k_m=1}^{K_m} \frac{\left|q_m(k_m)\right|^2}{p_m(k_m)}~\right]\,.
\end{align}
This variance is minimized if the sampling distributions $p_m$ are chosen to be proportional to $|q_m|$ as follows: %In order to minimize the variance of $W$, the sampling distributions $p_m$ are typically chosen to be proportional to $|q_m|$ as follows:
\begin{align}
 p^\text{(special-case)}_m(k_m) &\equiv \frac{\left|q_m(k_m)\right|}{\gamma_m}\,,\qquad\qquad \text{where}\quad \gamma_m\equiv \sum_{k_m=1}^{K_m} \left|q_m(k_m)\right|\,.
\end{align}
This is a popular choice for the sampling distribution; in this case, the weights $w_m(k_m)$ and $W(k_1,\dots,k_M)$ are given by
\begin{align}
 w^\text{(special-case)}_m(k_m) &= \sign\left(\Big.q_m(k_m)\right)\,\gamma_m\,,\\
 W^\text{(special-case)}(k_1,\dots,k_M) &= \left[\prod_{m=1}^M\sign\left(\Big.q_m(k_m)\right)\right]\,\gamma\,,\qquad\qquad\text{where}\quad \gamma\equiv \prod_{m=1}^M \gamma_m\,.
\end{align}
Here $\sign$ is the sign function. The corresponding weight-variance is given by
\begin{align}
 \var^\text{(special-case)}\left[W^\text{(special-case)}\right] &= \gamma^2 - \mu_W^2\,.
\end{align}
In this way, it can be seen that a larger value of $\gamma$ corresponds to a larger variance of $W$. The connection between $\gamma$ and the variance of $WX$ can be seen easily for the special case when $X$ is a random variable with a constant magnitude, say $a$, i.e., if $X=X_\text{(constant-mod)}\in \{-a, a\}$. In this case, we have
\begin{align}
 \var^\text{(special-case)}\left[W^\text{(special-case)}X_\text{(constant-mod)}\right] &= \gamma^2\,a^2 - T^2\,.
\end{align}
This demonstrates that a larger value of $\gamma$ corresponds to a larger variance of $WX$, and thus a larger sampling overhead. This feature generally holds true a) for both the inter- and intra-operation variance terms, and b) even when $X$ does not have a constant magnitude. Note that similar to $N_{\Pi K}$, $\gamma$ also scales exponentially in $M$.\footnote{As an aside, for a given value of $N_{\Pi K}$, if $\gamma$ is sufficiently large, then the Monte Carlo estimation technique may not be more advantageous than performing the summation explicitly.} As a result, the sampling overhead grows exponentially, e.g., with circuit-depth in PEC, or with the number of wire- and gate-cuts in circuit cutting.

The value of $\gamma$ depends only on the values of the coefficients $q_m(k_m)$. As a result it is a useful performance metric to compare different quasiprobability decompositions of the same target quantum operation, for example as in Refs.~\cite{Piveteau2022,10304329}, which is concerned with creating decompositions with reduced sampling overhead. However in this paper we do not create new decompositions. Rather we develop new estimators for $T$ with lower variance than $\Test{basic}$, using a given decomposition. So, instead of $\gamma$-factors, we will evaluate performance in terms of a) degree of variance reduction or, equivalently, b) the number of events needed to reach a given precision. Furthermore, we will not assume the particular choice for sampling distributions used in this section, namely $p_m = p_m^\text{(special-case)}$ in the development of our technique (however we will use this form for $p_m$ in our simulation experiments).

% while we will use the particular choice for sampling distributions used in this section, namely $p_m = p_m^\text{(special-case)}$ in our simulation experiments in \sref{sec:experiments}, we do not assume this form for $p_m$ in the development of our technique.

The rest of this paper is organized as follows. In \sref{sec:cv_intro}, we briefly review the control variates variance reduction technique. In \sref{sec:cv_construction} to \sref{sec:weight_as_control}, we discuss the construction of control variates for QPD estimations. In \sref{sec:cv_technique_modifications}, we introduce a few adaptations of the standard control variates technique, appropriate for QPD estimations. In \sref{sec:formulas}, we present formulas for algorthmically performing QPD estimations with control variates  . In \sref{sec:experiments}, we apply our techniques to realistic PEC tasks and quantitatively evaluate the performance of our techniques, before providing some conclusions and outlook in \sref{sec:conclusions}.

\section{Methodology} \label{sec:methodology}
\subsection{Introduction to the Control Variates Technique} \label{sec:cv_intro}
Control Variates techniques \cite{doi:https://doi.org/10.1002/9781118445112.stat07947} is a popular Monte Carlo variance reduction technique. The basic idea is as follows: Let $U$ be a random variable whose expected value $\E\left[U\right]$ we are interested in estimating, and let $V$ be a random variable whose expected value $\mu$ is known a priori. Now, one can create a one-parameter family of random variables, say $U_\lambda$, given by
\begin{align}
 U_\lambda \equiv U - \lambda (V - \mu)\,,\qquad\qquad\qquad\forall \lambda\in\mathbb{R}\,,
\end{align}
such that $\E[U_\lambda] = \E[U]$ for all $\lambda \in \mathbb{R}$. Here $V$ is referred to as a control variate. Although $U_\lambda$ has the same expectation as $U$, it could have a different variance. Concretely, using Bienaym\'{e}'s identity \cite{Klenke2013}, the variance of $U_\lambda$ can be written as
\begin{align}
 \var[U_\lambda] = \var[U] + \lambda^2 \var[V] - 2\,\lambda\,\cov[U, V]\,,
\end{align}
where $\cov[U, V]$ represents the covariance of $U$ and $V$. $\var[U_\lambda]$ is minimized by setting the coefficient $\lambda$ to $\lambda^\ast$ given by
\begin{align}
 \lambda^\ast \equiv\begin{cases}
  \frac{\displaystyle\cov[U, V]}{\displaystyle\var[V]}\,,&\qquad\qquad\text{if } \var[V] \neq 0\,,\\
  \quad 0\,,&\qquad\qquad\text{otherwise}\,.
 \end{cases}\label{eq:opt_lambda_1cv}
\end{align}
The corresponding minimum variance is given by
\begin{align}
 \var[U_{\lambda^\ast}] &= \begin{cases}
  ~~\var[U] - \frac{\displaystyle\cov^2[U, V]}{\displaystyle\var[V]} = \var[U]\left(1-\rho_{(U,V)}^2\right)\,,&\qquad\quad\text{if } \var[V] \neq 0\,,\\[.5em]
  ~~\var[U]\,,&\qquad\quad\text{otherwise}\,.
 \end{cases}\label{eq:min_var_1cv}
\end{align}
where $\rho_{(U,V)}$ is the Pearson correlation coefficient of $U$ and $V$. Thus, if $V$ is correlated to $U$, \eqref{eq:min_var_1cv} provides a way to construct a random variable with the same expectation as $U$ but a lower variance, thereby ameliorating the sampling overhead---the stronger the correlation between $U$ and $V$, the greater the variance reduction. In the context of this paper, $U$ will be identified with $WX$, the random variable whose mean is to be estimated. We will postpone the discussion of suitable control variates to \sref{sec:cv_construction}.

The control variates technique can also be generalized to include multiple control variates, say $V_1,\dots, V_{N_\mathrm{cv}}$, with means $\mu_1,\dots, \mu_{N_\mathrm{cv}}$, respectively. In this case, one can create an $N_\mathrm{cv}$-parameter family of random variables $U_{(\lambda_1,\dots, \lambda_{N_\mathrm{cv}})}$ given by
\begin{align}
 U_{(\lambda_1,\dots, \lambda_{N_\mathrm{cv}})} \equiv U - \sum_{a=1}^{N_\mathrm{cv}} \lambda_a\,\left(V_a-\mu_a\right)\,,\qquad\qquad\qquad\forall \lambda_1,\dots,\lambda_{N_\mathrm{cv}}\in\mathbb{R}\,,\label{eq:multiple_cvs}
\end{align}
with the same expectation as $U$. Again using Bienaym\'{e}'s identity, the variance of $U_{(\lambda_1,\dots, \lambda_{N_\mathrm{cv}})}$ can be written as
\begin{align}
 \var\left[U_{(\lambda_1,\dots, \lambda_{N_\mathrm{cv}})}\right] = \var[U] + \sum_{a=1}^{N_\mathrm{cv}} \lambda_a\,\lambda_b\,K_{ab} - 2\,\sum_{a=1}^{N_\mathrm{cv}}\lambda_a\,\cov[U, V_a]\,,\label{eq:var_U_lambdas}
\end{align}
where $K_{ab}$ is the the covariance of $V_a$ and $V_b$ or, equivalently, the $(a,b)$-th element of the covariance matrix $\mathbf{K}$ of the control variates. $\var\left[U_{(\lambda_1,\dots, \lambda_{N_\mathrm{cv}})}\right]$ is minimized by setting $\lambda_a$-s to $\lambda^\ast_a$-s given by
\begin{align}
 \lambda^\ast_a = \sum_{b=1}^{N_\mathrm{cv}} K^+_{ab}\,\cov\left[U, V_b\right]\,,\label{eq:lambda_star}
\end{align}
where $K^+_{ab}$ is the $(a,b)$-th element of the Moore--Penrose inverse (or the pseudoinverse) $\mathbf{K}^+$ of the covariance matrix $\mathbf{K}$. The corresponding minimum variance is given by
\begin{align}
 \var\left[U_{(\lambda^\ast_1,\dots, \lambda^\ast_{N_\mathrm{cv}})}\right] &= \var[U] - \sum_{a,b=1}^{N_\mathrm{cv}} K^+_{ab}\,\cov[U, V_a]\,\cov[U, V_b]\,.\label{eq:var_U_lambdas_opt}
\end{align}
The variance of $U_{(\lambda^\ast_1,\dots, \lambda^\ast_{N_\mathrm{cv}})}$ can roughly be thought of as the variance in $U$ that is \emph{unexplained} by \cite{Achen_1990} or unaccounted for by the control variates. As an advantageous special case, if $(U-\E[U])$ lies in the span of the random variables $(V_a-\mu_a)$, i.e., if $(U-\E[U])$ can be written as a linear combination of $(V_a-\mu_a)$-s, then $\var\left[U_{(\lambda^\ast_1,\dots, \lambda^\ast_{N_\mathrm{cv}})}\right]$ will be zero (stated as \lref{lemma:span} and proved in \aref{appendix:lemmas_and_proofs}).

Note that the values of $\cov[U, V_a]$ and/or $K^+_{ab}$ may not be known a priori. In this case, $\lambda^\ast_a$ cannot be computed exactly, and good values for the coefficients $\lambda_a$ would have to be estimated from the data. This estimation process for the specific task at hand, namely sampling overhead reduction in QPD, will be discussed in \sref{sec:cv_technique_modifications}.

\subsection{Constructing Control Variates} \label{sec:cv_construction}
In order to use the control variates technique for QPD sampling overhead reduction, we need random variables with a priori known (i.e., pre-computable) expectations to use as controls. Note that the weight $W$ itself can be used as a control, with its mean $\mu_W$ pre-computed using \eqref{eq:w_mean}. Likewise, the constituent weights $w_m$ can also be used as controls, with their means pre-computed using \eqref{eq:wm_mean}. More generically, one can create a random variable $V_\mathrm{fac}$, which is a factorizable function of $(k_1,\dots,k_M)$, as follows:
\begin{align}
 V_\mathrm{fac}(k_1,\dots, k_M) &\equiv \prod_{m=1}^M v_m(k_m)\,. \label{eq:v_fac}
\end{align}
Here each $v_m$ is an arbitrary real-valued function of $k_m\in\{1,\dots, K_m\}$. $V_\mathrm{fac}$ can be used as a control, with its expected value pre-computed as follows:
\begin{align}
 \E\left[V_\mathrm{fac}\right] &= \prod_{m=1}^M \E\left[v_m\right]\,,\quad\text{where}\quad \E\left[v_m\right] = \sum_{k_m=1}^{K_m} p_m(k_m)\,v_m(k_m)\,. \label{eq:fac_expectation}
\end{align}
The full weight $W$ as well as the constituent weights $w_m$ are all special cases of the generic factorizable control variate $V_\mathrm{fac}$. In addition to factorizable functions like $V_\mathrm{fac}$, their linear combinations, e.g., $a\,V_\mathrm{fac} + b\,V'_\mathrm{fac}$, can also be used as control variates. The mean of the linear combination will simply be the linear combination of the means, which can be pre-computed easily as long as the number of terms in the linear combination is practicable. Note that by constructing control variates as functions of $k_m$-s, we are facilitating correlations between the controls and $WX$ (which itself is influenced by the $k_m$-s). Such correlations are crucial for variance reduction using the control variates technique. Also, because these control variates are deterministic functions of $(k_1,\dots,k_M)$, they can only reduce the effect of the inter-operation variance term in \eqref{eq:variance_decomposition}, and not the effect of the intra-operation variance term.

We end this subsection by noting that in addition to the expectations, other statistical properties like covariance can also be pre-computed for control variates of the forms discussed in this subsection. As an illustration, for pairs of factorizable control variates of the form in \eqref{eq:v_fac}, say $V_\mathrm{fac}$ and $V'_\mathrm{fac}$, the covariance can be pre-computed using the formula
\begin{align}
 \cov\left[V_\mathrm{fac}\,,\,V'_\mathrm{fac}\right] \equiv \E\left[V_\mathrm{fac}\,V'_\mathrm{fac}\right] - \E\left[V_\mathrm{fac}\right]\,\E\left[V'_\mathrm{fac}\right]\,,
\end{align}
by noting that the product $V_\mathrm{fac}\,V'_\mathrm{fac}$ is also a factorizable function of the $k_m$-s. For a pair of control variates, each of which is a linear combination of factorizable control variates, one can compute the covariance using the distributive property of covariance over addition, as long as the numbers of terms in the linear combinations are practicable.

\subsection{Existence of a ``Golden-Ticket'' Control Variate} \label{sec:golden-control}
As a special case of the $V_\mathrm{fac}$ in \eqref{eq:v_fac}, one can define indicator variables $I_{(k'_1,\dots,k'_M)}$ given by
\begin{alignat}{2}
 I_{(k'_1,\dots,k'_M)}(k_1,\dots,k_M) &\equiv \begin{cases}
  ~~1\,,\qquad\qquad&\text{if }k_m=k'_m\,\quad\forall m\in\{1,\dots,M\}\,,\\
  ~~0\,,\qquad\qquad&\text{otherwise}\,,
 \end{cases}\label{eq:indicator}%\\[1em]
 % I_{(m, k')}(k_1,\dots,k_M) &\equiv \begin{cases}
 %  ~~1\,,\qquad\qquad&\text{if }k_m=k'\,,\\
 %  ~~0\,,\qquad\qquad&\text{otherwise}\,,
 % \end{cases}
\end{alignat}
with expectations given by
\begin{align}
 \E\left[I_{(k'_1,\dots,k'_M)}\right] &= \prod_{m=1}^M p_m(k'_m)\,.%\\
 % \E\left[I_{(m, k')}\right] &= p_m(k')\,.
\end{align}
The indicator functions $I_{(k'_1,\dots,k'_M)}$ form a basis for the space of all possible deterministic functions of $(k_1,\dots,k_M)$. This means that \emph{any deterministic function} of $(k_1,\dots,k_M)$ can be written as a linear combination of factorizable functions of the form in \eqref{eq:v_fac}. A corollary is that the random variable
\begin{align}
 V_\mathrm{golden} \equiv W\,\E\left[X\,\big|\,k_1,\dots, k_M\right]\,,
\end{align}
is expressible as a linear combination of $V_\mathrm{fac}$-s.
% A corollary is that there exists \emph{some}\footnote{$V_\mathrm{golden}$ is unique up to scaling and shifting.} control variate, say $V_\mathrm{golden}$, expressible as a linear function of $V_\mathrm{fac}$-s, which is perfectly correlated with $W\,\E\left[X\,\big|\,k_1,\dots, k_M\right]$.
Roughly speaking, this a priori unknown ``golden-ticket'' variable, if used as a control, can completely eliminate the inter-operation variance of $WX$. More precisely, the variance of $U_{\lambda^\ast}$ in \eqref{eq:min_var_1cv}, with $U$ set to $WX$ and $V$ set to $V_\mathrm{golden}$,\footnote{A control that is perfectly correlated with $V_\mathrm{golden}$ will be just as effective.} will only be the intra-operation part of the variance of $WX$, given by
\begin{align}
 \E\left[\Big.W^2\,\var\left[X\,\big|\,k_1, \dots, k_M\right]\right]\,;
\end{align}
this is stated as \lref{lemma:goldenticket} and proved in \aref{appendix:lemmas_and_proofs}. However, finding such a special control variate may not be easy, given that the dimensionality of the relevant function space is $N_{\Pi K}$, typically a very large number.

\subsection{The Effectiveness of Weight as a Control} \label{sec:weight_as_control}
As noted above, the dimensionality of the space of deterministic functions of $(k_1,\dots,k_M)$ could be very large. As such, if one picks a control variate randomly from this high-dimensional function space (using some arbitrary sampling distribution), one would expect the correlation between the control and $WX$ to be weak, and consequently, the control variate to be ineffective for variance reduction. However, it is possible to construct effective controls, as will be demonstrated in this paper. % by exploiting certain features of the problem at hand. We discuss one such feature here and another in \sref{}.

Let us consider the performance of the weight $W$ as a control. For simplicity, let us consider the scenario where $W\in \{+\gamma, -\gamma\}$, $\mu_W=1$, and $X\in \{+1, -1\}$. By construction, the quantities $\E\left[X\right]$, $\E\left[X\,\big|\,W=+\gamma\right]$, and $\E\left[X\,\big|\,W=-\gamma\right]$ all lie in the range $[-1, 1]$. In addition, as is typical, let the target $T\equiv \E[WX]$ also lie in the range $[-1, 1]$. In this setup, the optimal coefficient $\lambda^\ast$ from \eqref{eq:opt_lambda_1cv} (for $U\equiv WX$, $V\equiv W$) and $\rho^2_{(W, WX)}$, and their limits as $\gamma\rightarrow\infty$, are given by
\begin{alignat}{2}
 \lambda^\ast &= \frac{\gamma^2\,\E\left[X\right] - \E\left[WX\right]}{\gamma^2 - 1}\,,\qquad\qquad&&\lambda^\ast \longrightarrow \E\left[X\right]\quad\text{as}\quad \gamma\longrightarrow \infty\,, \label{eq:W_lambda_lim}\\
 \rho^2_{(W, WX)} &= \frac{\displaystyle\left(\gamma^2\,\E\left[X\right] - \E\left[WX\right]\right)^2}{\displaystyle\left(\gamma^2-1\right)\,\left(\gamma^2 - \E^2\left[WX\right]\right)}\,,\qquad\qquad&&\rho^2_{(W, WX)} \longrightarrow \E^2\left[X\right]\quad\text{as}\quad \gamma\longrightarrow \infty\,.
\end{alignat}
$\rho^2_{(W, WX)}$ captures the fraction by which the control variate $W$ can reduce the variance of $WX$. For example, roughly speaking, a value of $0.4$ for $\rho^2_{(W,WX)}$ corresponds to a $40\%$ reduction in the number of datapoints needed to achieve a given precision in the estimate for $T$. $\rho^2_{(W, WX)}$ is plotted, as a function of $\E\left[X\right]$ and $\E\left[WX\right]$, for a few different values of $\gamma$, in the heatmaps in \fref{fig:W_theoretical_performance}. The hatched regions in the heatmaps correspond to values of $\E\left[X\right]$ and $\E\left[WX\right]$ that are disallowed because they correspond to $\E\left[X\,\big|\,W=-\gamma\right]$ having a magnitude greater than 1. Qualitatively, we can see that representative values of $\rho^2_{(W, WX)}$ in the plots correspond to a reasonably effective variance reduction (we postpone a quantitative evaluation of the performance to \sref{sec:experiment_results}). Furthermore, this performance does not degrade as $\gamma$ and/or $N_{\Pi K}$ increases, unlike what could be expected of the performance of randomly picked (uninformed) control variates.
\begin{figure}
 \centering
 \includegraphics[width=\textwidth]{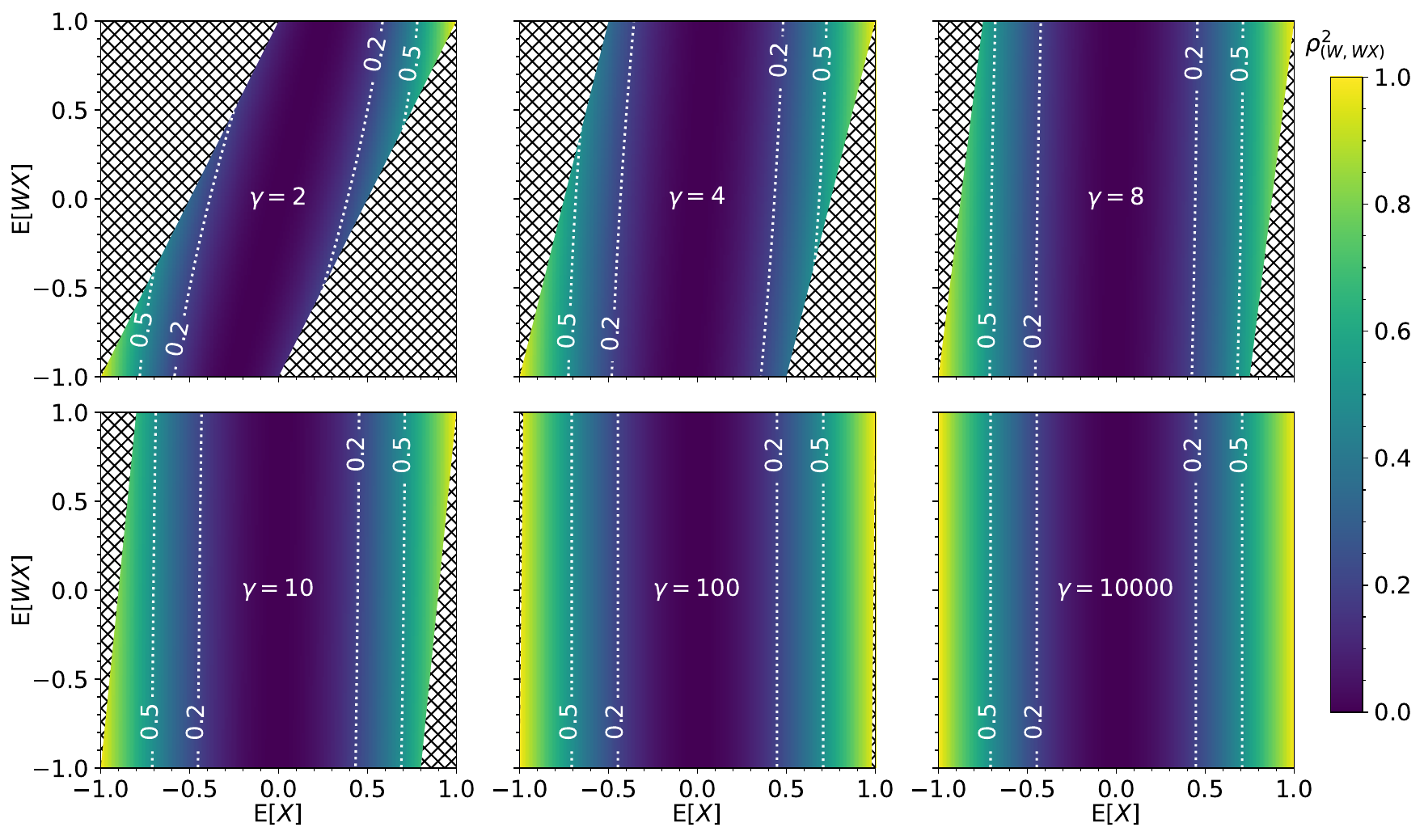}
 \caption{A heatmap of $\rho^2_{(W,WX)}$ as a function of $\E[W]$ and $\E[WX]$ for different values of $\gamma$, in the case where $W\in\{+\gamma, -\gamma\}$, $\E[W]=1$, and $X\in\{+1, -1\}$. The hatched regions in the heatmaps correspond to $\left|\big.\E\left[X\,|\,W=-\gamma\right]\right| > 1$, which is not allowed.}
 \label{fig:W_theoretical_performance}
\end{figure}

Leveraging the effectiveness of $W$ as a control, other effective control variates could be designed as slight modifications or variations of $W$. Examples of such a constructions are given in \sref{sec:experiments_cvs}. %\textcolor{red}{Leveraging Disallowed $\E[X|\gamma\pm]$ in outlook}

\subsection{Adapting the Standard Control Variates Technique} \label{sec:cv_technique_modifications}
Having described a method to construct control variates for QPD-estimations, next we will turn to the task of calculating good values for the coefficients $\lambda_a$ in \eqref{eq:multiple_cvs}. From \eqref{eq:lambda_star}, the optimal coefficients $\lambda^\ast_a$ for $U\equiv W\,X$ is given by
\begin{align}
 \lambda^\ast_a &= \sum_{b=1}^{N_\mathrm{cv}} K^+_{ab}\,\cov\left[W\,X, V_b\right] \label{eq:lambda_star_generic_form}\\
 &= \sum_{b=1}^{N_\mathrm{cv}} K^+_{ab}\,\cov\left[\Big.W\,\left(\big.X-\E\left[X\right]\right), V_b\right] + \sum_{b=1}^{N_\mathrm{cv}} K^+_{ab}\,\cov\left[\Big.W\,\E\left[X\right], V_b\right]\\
 &= \sum_{b=1}^{N_\mathrm{cv}} K^+_{ab}\,\cov\left[\Big.X, W\,\left(\big.V_b-\mu_b\right)\right] + \E\left[X\right]\,\sum_{b=1}^{N_\mathrm{cv}} K^+_{ab}\,C_b\,, \label{eq:lambda_star_special_form}
\end{align}
where $C_a$ is the covariance of $W$ and $V_a$:
\begin{align}
 C_a \equiv \cov\left[W, V_a\right]\,.
\end{align}
In general the value of $\lambda_a^\ast$ is not known a priori. In such situations, the general strategy is to use a data-based-estimate for the $\lambda^\ast_a$-s as the coefficients for the control variates. Explicitly, the corresponding control-variate-based estimator for $T$ would be given by
\begin{align}
 \Test{standard-cv} &\equiv \frac
{1}{N}\sum_{i=1}^N \left[W^{(i)}\,X^{(i)} - \sum_{a=1}^{N_\mathrm{cv}}\widehat{\lambda}^\ast_a \left(V_a^{(i)}-\mu_a\right)\right]
\end{align}
Here, $V_a^{(i)}$ is the value of the $a$-th control variate for the $i$-th datapoint. $\widehat{\lambda}^\ast_a$-s are data-based-estimates for $\lambda^\ast_a$---they are functions of the full dataset, which includes the values of $X^{(i)}$, $W^{(i)}$, and $V_a^{(i)}$ for all $i\in\{1, \dots, N\}$ and $a\in\{1,\dots,N_\mathrm{cv}\}$. Here we will modify this standard approach in the following ways.

\subsubsection{Eliminating Bias} Since the coefficients $\widehat{\lambda}^\ast_a$ are functions of the random data, they are themselves random variables that could be correlated with the values of $V_a^{(i)}$. As a result, $\Test{standard-cv}$ is not an unbiased estimator for $T$. As the number of datapoints used to estimate $\widehat{\lambda}^\ast_a$ increases, the bias in $\Test{standard-cv}$ will reduce in magnitude, since the correlation between $\widehat{\lambda}^\ast_a$ and $V_a^{(i)}$ for any given $i$ reduces in magnitude. In many applications of the control variates technique, the number of datapoints is large enough for the bias to be negligible. However, this may not be the case in quantum computing applications, so we modify
%the standard control variate based estimator
$\Test{standard-cv}$
as follows to eliminate the bias:
\begin{align}
 \Test{unbiased-cv} &\equiv \frac
{1}{N}\sum_{i=1}^N \left[W^{(i)}\,X^{(i)} - \sum_{a=1}^{N_\mathrm{cv}}\lambda^{(\neg i)}_a \left(V_a^{(i)}-\mu_a\right)\right]\,, \label{eq:Test_unbiased_cv}
\end{align}
where $\lambda^{(\neg i)}_a$ is an estimate for $\lambda^\ast_a$ that uses all the datapoints except the $i$-th one. Each datapoint $i$ will have a different value for $\lambda^{(\neg i)}_a$. Since $\lambda^{(\neg i)}_a$ is independent of $V_a^{(i)}$, $\Test{unbiased-cv}$ will be unbiased, i.e., $\E\left[\Test{unbiased-cv}\right] = T$. There are other ways to eliminate the bias, e.g., a) splitting the dataset into two halves, and for each half of the dataset, using the value of $\widehat{\lambda}^\ast_a$ estimated from the other half, b) using a small fraction of datapoints \emph{only} for estimating $\widehat{\lambda}^\ast_a$. However, we avoid such approaches here because they a) lead to estimators that are not invariant under permutations of the dataset, and b) only use a small fraction of the dataset (which is likely to be small to begin with in near-term quantum applications) to estimate the coefficients.

\subsubsection{Enforcing Shift Invariance}
Under the transformation $X\mapsto X+\delta$, it is easy to see that $T$ transforms as $T\mapsto T+\mu_W\,\delta$. This motivates the definition of a shift invariance property for estimators of $T$. Let us define an estimator $\widehat{T}$ for $T$ as being shift invariant if, under the transformation $X^{(i)} \mapsto X^{(i)} + \delta$, the estimator transforms as $\widehat{T} \mapsto \widehat{T} + \mu_W\,\delta$. Somewhat surprisingly, the basic Monte Carlo estimator $\Test{basic}$ is not a shift invariant estimator\footnote{This lack of shift-invariance was the original clue which pointed us towards the use of control variates for this problem.}, since it transforms as
\begin{align}
 \Test{basic} \mapsto \Test{basic} + \frac{\delta}{N}\sum_{i=1}^N W^{(i)}\not\equiv \Test{basic} +\mu_W\,\delta\,.
\end{align}
If one uses an estimator that is not shift invariant, the estimated value of $T$ would be different if one simply shifts $X$ by $\delta\neq 0$ before performing the QPD estimation (and applies the appropriate inverse-shift of $-\mu_W\,\delta$ after the estimation). This is arguably inelegant, so shift-invariance is a desirable property for estimators of $T$; we will design our control-variate-based estimators to satisfy this property. $\Test{unbiased-cv}$ will be shift invariant if the following two conditions are both satisfied (proved as \lref{lemma:shiftinvariance} in \aref{appendix:lemmas_and_proofs}):
\begin{enumerate}
 \item $(W-\mu_W)$ lies in the span of the random variables $(V_a-\mu_a)$. 
 \item $\lambda^{(\neg i)}_a$ transforms as $\displaystyle\lambda^{(\neg i)}_a \mapsto \lambda^{(\neg i)}_a + \delta \sum_{b=1}^{N_\mathrm{cv}}K^+_{ab}\,C_b$ under $X^{(i)}\mapsto X^{(i)}+\delta$.
\end{enumerate}
The first condition can be enforced simply by including $W$ (or a variable proportional to $W$) as one of the controls. In order to enforce the second condition, we will use the expression for $\lambda^\ast$ in \eqref{eq:lambda_star_special_form} to estimate it as follows:
\begin{align}
 \lambda^{(\neg i)}_a &\equiv \sum_{b=1}^{N_\mathrm{cv}}K^+_{ab}\,S_b^{(\neg i)} + \left(\frac{1}{N-1}\sum_{j\neq i} X^{(j)}\right)\sum_{b=1}^{N_\mathrm{cv}} K^+_{ab}\,C_b\,,\label{eq:lambda_a_sans1}
\end{align}
where $S_b^{(\neg i)}$ is the sample covariance of $X$ and $W(V_b-\mu_b)$ computed using all the datapoints except the $i$-th one. Note that the values of $K_{ab}$ and $C_b$ can be pre-computed for the control variates constructed and used in this paper, as discussed in \sref{sec:cv_construction}. 
The Moore--Penrose inverse of $\mathbf{K}$, namely $\mathbf{K}^+$, can also be pre-computed, using the  singular value decomposition of $\mathbf{K}$; this is implemented in many standard computational tools, e.g., Numpy package \cite{harris2020array} for the Python language \cite{van1995python}.
%The values of $K^+_{ab}$, which are the elements of the Moore--Penrose inverse of $\mathbf{K}$, can also be pre-computed using standard computational tools, e.g., Numpy package \cite{} with the Python language \cite{}.

\subsubsection{Error Estimation}
The basic Monte Carlo estimator $\Test{basic}$ is the sample mean of $N$ independent and identically distributed random variables $W^{(i)}\,X^{(i)}$. In this case, one can estimate the error in $\Test{basic}$, e.g., as $(1/\sqrt{N})$ times the sample standard deviation of $W^{(i)}\,X^{(i)}$. On the other hand, because of the data dependence of the coefficients $\lambda^{(\neg i)}_a$, the control-variate-based estimators are sample means of identically distributed and \emph{dependent} (possibly correlated) random variables. In such situations, the aforementioned error-estimation-technique could underestimate the uncertainty of the estimator. The degree of this underestimation decreases as the number of datapoints $N$ increases---if $N$ is sufficiently large this problem can be ignored. However, to cover the case where $N$ is not very large, in \sref{sec:formulas}, we provide some error estimation formulas pertaining to our Cv4Quantum estimators. These formulas, which are derived in \aref{appendix:uncertainty_estimates}, include higher order correction terms to provide a better uncertainty estimate.%, as discussed in \sref{}.

\subsection{Choosing the Number of Controls} \label{sec:num_controls}
Adding more controls to a given set of control variates cannot increase the minimum variance $\var\left[U_{\lambda_1^\ast,\dots,\lambda_{N_\mathrm{cv}}^\ast}\right]$ in \eqref{eq:var_U_lambdas_opt}. In fact, under some weak conditions, adding more controls will strictly decrease $\var\left[U_{\lambda_1^\ast,\dots,\lambda_{N_\mathrm{cv}}^\ast}\right]$. This suggests that one should use as many control variates as is computationally feasible in order to improve the variance reduction achieved. However, note that in practice one does not have access to the optimal coefficients $\lambda_a^\ast$, and instead one uses coefficients estimated from data; using poorly estimated coefficients can lead to a poor estimation of $T$. Since one only has access to a finite-sized dataset, increasing the number of controls indefinitely will naturally lead to poorly estimated coefficients and, consequently, worsening of the errors in the estimate of $T$ (this is reflected in the error estimation formula in \sref{sec:formulas}, \alref{alg:cv_estimator}).

The prudent approach is to find a good compromise between the potential improvement from incorporating more controls and the deterioration resulting from the poorer estimation of the corresponding coefficients. The number of controls $N_\mathrm{cv}$ one uses should depend on the number of data points $N$ used in the estimation---a larger value of $N$ should allow for a larger value of $N_\mathrm{cv}$. Providing a more concrete prescription for how to choose the number of controls is beyond the scope of this paper.

\section{CV4Quantum: Formulas for Estimation} \label{sec:formulas}
%\subsection{Definitions and Notation}
We dub the control-variates-based QPD estimation technique of this paper as ``CV4Quantum''. To facilitate the adoption of CV4Quantum, the information in \sref{sec:cv_intro}, \sref{sec:cv_technique_modifications}, and \aref{appendix:uncertainty_estimates} is distilled, in this section, into a set of formulas which can be used directly to implement our estimators. \Tref{tab:definitions} summarizes the descriptions of quantities which are provided as inputs to the estimators.
\begin{table}[ht]
 \centering
 \caption{The definitions and descriptions of quantities that are provided as inputs to the QPD estimators.}
 \label{tab:definitions}
 \resizebox{\textwidth}{!}{\begin{tabular}{M{.1\textwidth} m{.85\textwidth}}
  \toprule
   \textbf{Symbol} & \cellcenter{\textbf{Definition/Description}}\\
  \midrule[.8pt]
   % & \hskip 10em \underline{From experimental data}\\[1em]
   $N$ & Number of data points.\\
   \cmidrule(lr){1-2}
   $N_\mathrm{cv}$ & Number of control variates. The control variates are $V_1,\dots, V_{N_\mathrm{cv}}$.\\
   \cmidrule(lr){1-2}
   $X^{(i)}$ & Value of the computation-result $X$ for the $i$-the datapoint.\\
   \cmidrule(lr){1-2}
   $W^{(i)}$ & Value of the weight $W$ for the $i$-the datapoint.\\
   \cmidrule(lr){1-2}
   $V_a^{(i)}$ & Value of the control variate $V_a$ for the $i$-the datapoint.\\
   \cmidrule(lr){1-2}
   $\mu_W$ & Population mean (i.e., expected value) of the weight $W$. This should be pre-computed, and will typically be 1 in PEC.\\
   \cmidrule(lr){1-2}
   % $\sigma^2_W$ & Population variance of the weight $W$. This should be pre-computed, and will typically be $\gamma^2-1$ in PEC.\\
   % \cmidrule(lr){1-2}
   $\mu_a$ & Population mean of the control variate $V_a$ (to be pre-computed).\\
   \cmidrule(lr){1-2}
   $K_{ab}$ & Population covariance of control variates $V_a$ and $V_b$ (to be pre-computed).\newline Constitutes the $(a,b)$-th element of the covariance matrix $\mathbf{K}$.\\
   \cmidrule(lr){1-2}
   $K^+_{ab}$ & $(a,b)$-th element of $\mathbf{K}^+$, the Moore--Penrose inverse of $\mathbf{K}$ (to be pre-computed).\\
   \cmidrule(lr){1-2}
   $C_a$ & Population covariance of $W$ and $V_a$ (to be pre-computed). \\
 \bottomrule
 \end{tabular}}
\end{table}

We will use three different notations for data-indices, which run from $1$ to $N$, namely a)~superscript~$^{(i)}$, b)~superscript~$^{(\neg i)}$, and c)~superscript~$^{(\sim i)}$, with the following meanings: $A^{(i)}$ is a quantity that depends \textbf{only on} the $i$-th datapoint. $A^{(\neg i)}$ is a quantity that \textbf{does not} depend on the $i$-th datapoint. $A^{(\sim i)}$ is a generic quantity indexed by $i$, that may or may not depend on any given datapoint.%\footnote{Every quantity in this section with a data-index (of any of the three types), say $A^{(\sim i)}$, is defined so that permuting the datapoints will lead to a corresponding permutation in the values of $A^{(\sim i)}$.}

Next we introduce shorthand notations for the sample mean, sample variance, and sample covariance as follows:
\begin{align}
 \smean\left[A\right] &\equiv \frac{1}{N}\sum_{i=1}^N A^{(\sim i)}\\
 \svar\left[A\right] &\equiv \frac{1}{N-1}\sum_{i=1}^N \left(A^{(\sim i)} - \smean\left[A\right]\right)^2\\
 \scov\left[A, B\right] &\equiv \frac{1}{N-1}\sum_{i=1}^N \left[\bigg.\left(A^{(\sim i)} - \smean\left[A\right]\right)\left(B^{(\sim i)} - \smean\left[B\right]\right)\right]
\end{align}
Additionally, we introduce shorthand notations for the sample mean and sample covariance leaving one element out:
\begin{align}
 &\smeansans^{(\sim i)}\left[A\right] \equiv \frac{1}{N-1}\sum_{j\neq i} A^{(\sim j)} = \frac{N}{N-1}\,\smean\left[A\right] - \frac{1}{N-1}\,A^{(\sim i)}\\
 % \svarsans^{(\sim i)}\left[A\right] &\equiv \frac{1}{N-2}\sum_{j\neq i} \left(A^{(\sim j)} - \smeansans^{(\sim i)}\left[A\right]\right)^2\\
 % &= \frac{N-1}{N-2}\,\svar\left[A\right] - \frac{N}{(N-1)(N-2)}\left[A^{(\sim i)} - \smean\left[A\right]\right]^2\\
 &\scovsans^{(\sim i)}\left[A, B\right] \equiv \frac{1}{N-2}\sum_{j\neq i} \left[\bigg.\left(A^{(\sim j)} - \smeansans^{(\sim i)}\left[A\right]\right)\left(B^{(\sim j)} - \smeansans^{(\sim i)}\left[B\right]\right)\right]\\
 &\qquad\qquad= \frac{N-1}{N-2}\,\scov\left[A, B\right] - \frac{N}{(N-2)(N-1)}\left(A^{(\sim i)} - \smean\left[A\right]\right)\left(B^{(\sim i)} - \smean\left[B\right]\right)
\end{align}
If the superscripts of $A$ and $B$ are of the $^{(i)}$ type, then the superscripts of $\smeansans$ and $\scovsans$ can be changed to the $^{(\neg i)}$ type. Note that, starting from the values of $A^{(\sim i)}$ and $B^{(\sim i)}$, the values of $\smeansans^{(\sim i)}$ and $\scovsans^{(\sim i)}$ can be computed for all $i\in\{1, \dots, N\}$ in $\mathcal{O}(N)$ time. Finally, we define $R^{(\sim i)}$ as follows:
\begin{align}
 R^{(\sim i)} &\equiv X^{(i)} - \smean\left[X\right]\,.
\end{align}
By construction $R^{(\sim i)}$ a) is invariant under constant shifts of $X$, and b) has a sample mean of 0.

Using these preliminary definitions and shorthand notations, we will now provide concrete definitions for a few estimators of $T$. Each definition is a sequence of equations which describe an algorithm to compute $\Test{}$, which is an estimate for $T$, and $\sigmaest{}^2$, which is an estimate for $\var\left[\Test{}\right]$. For completeness, the computation of the basic Monte Carlo estimator discussed in \sref{sec:monte_carlo_intro} is outlined in \alref{alg:basic_estimator}.

% \subsection{Basic Monte Carlo estimator}
\begin{algorithm}
\begin{subequations}
% \begin{empheq}[box=\widefbox]{align}
\begin{align}
 Y^{(i)}_\text{basic} &\equiv W^{(i)}\,X^{(i)}\,,\\
 \Test{basic} &\equiv \smean\left[Y_\text{basic}\right]\,,\\
 \sigmaest{basic}^2 &\equiv \frac{1}{N}\,\svar\left[Y_\text{basic}\right]\,.
 % \notag\text{where}\qquad\qquad\\
\end{align}
% \end{empheq}
\end{subequations}
\caption{\small Formulas to compute $\widehat{T}_\text{basic}$, the basic Monte Carlo estimator for $T$, and $\widehat{\sigma}_\text{basic}^2$, an estimate for the variance of $\widehat{T}_\text{basic}$.}
\label{alg:basic_estimator}
\end{algorithm}

In \alref{alg:centered_estimator}, we provide an estimator for $T$ that uses $W$ as a control, with the coefficient $\lambda^{(\neg i)}$ chosen to be $\smeansans^{(\neg i)}\left[X\right]$. The motivations behind this choice are a) the $\gamma\rightarrow\infty$ limit for $\lambda^\ast$ in \eqref{eq:W_lambda_lim} and b) making the estimator shift-invariant. The definition of $\Test{centered}$ in terms of $Y_\mathrm{centered}$ shows that $\Test{centered}$ is an unbiased estimator for $T$, while the definition in terms of $Z_\mathrm{centered}$ a) is numerically more stable, and b) demonstrates the shift invariance property of the estimator.
% \subsection{Centered estimator}
\begin{algorithm}
\begin{subequations}
% \begin{empheq}[box=\widefbox]{align}
\begin{align}
 Z^{(\sim i)}_\mathrm{centered} &\equiv \left[\frac{N\,W^{(i)} - \mu_W}{N-1}\right]\,R^{(\sim i)}\,,\\
 Y^{(\sim i)}_\mathrm{centered} &\equiv W^{(i)}\,X^{(i)} - \smeansans^{(\neg i)}\left[X\right]\left(W^{(i)}-\mu_W\right)\,,\label{eq:alg2_b}\\
 &= Z^{(\sim i)}_\mathrm{centered} + \mu_W\,\smean\left[X\right]\,,\\
 \Test{centered} &\equiv \smean\left[Y_\mathrm{centered}\right] = \smean\left[Z_\mathrm{centered}\right] + \mu_W\,\smean\left[X\right]\,,\\
 % \hat{\sigma}^2_{\mathrm{centered}\,,1} &\equiv \frac{1}{N}\,\svar\left[Y_\mathrm{centered}\right] = \frac{1}{N}\,\svar\left[Z_\mathrm{centered}\right]\,,\\
 % \hat{\sigma}^2_{\mathrm{centered}\,,2} &\equiv \hat{\sigma}^2_{\mathrm{centered}\,,1} + \frac{1}{(N-1)^2}\,\scov^2\left[X, W\right]\,,\\
 \sigmaest{centered}^2 &\equiv \frac{1}{N}\,\svar\left[Y_\mathrm{centered}\right] + \frac{1}{(N-1)^2}\,\scov^2\left[X, W\right] \label{eq:alg2_e}\\
 &= \frac{1}{N}\,\svar\left[Z_\mathrm{centered}\right] + \frac{1}{(N-1)^2}\,\scov^2\left[X, W\right]\,.
\end{align}
% \end{empheq}
\end{subequations}
\caption{\small Formulas to compute $\widehat{T}_\text{centered}$ (the centered Monte Carlo estimator for $T$) and $\widehat{\sigma}_\text{centered}^2$ (an estimate for the variance of $\widehat{T}_\text{centered}$).}
\label{alg:centered_estimator}
\end{algorithm}

% \subsection{Control-Variate-Based Estimator}
In \alref{alg:cv_estimator}, we provide a sequence of formulas for computing the control-variates-based estimate of $T$, which is one of the main contributions of this paper. The definition of $\Test{cv}$ via \eqref{eq:alg3_c}, \eqref{eq:alg3_d}, \eqref{eq:alg3_f}, and \eqref{eq:alg3_h} follows from plugging the expression for $\lambda^{(\neg i)}_a$ in \eqref{eq:lambda_a_sans1} into \eqref{eq:Test_unbiased_cv}. The formula for $\sigmaest{cv}$ is derived in \aref{appendix:uncertainty_estimates}. The definition of $\Test{cv}$ in \alref{alg:cv_estimator} in terms of $Y_\mathrm{cv}$ shows that $\Test{cv}$ is an unbiased estimator for $T$, while the definition in terms of $Z_\mathrm{cv}$ a) is numerically more stable, and b) demonstrates the behavior of the estimator under a constant shift in the values of $X^{(i)}$. In \eqref{eq:alg3_i} for $\sigmaest{cv}$, the second term on the right-hand side captures the error contribution from the fact that the coefficients $\lambda^{(\neg i)}_a$ are estimated from data. Note that this contribution to the variance increases with $N_\mathrm{cv}$, the number of coefficients to be estimated.\footnote{This term scales roughly linearly with $N_\mathrm{cv}$ (which becomes apparent from \eqref{eq:alg3_i} when $\mathbf{K}$ is diagonalized), provided a few other statistical properties of the control variates stay the same as more controls are added} This is related to the discussion in \sref{sec:num_controls} regarding choosing the number of controls.
\begin{algorithm}
\begin{subequations} \label{eq:cv_alg}
% \begin{empheq}[box=\widefbox]{align}
\begin{align}
 G_a^{(i)} &\equiv W^{(i)}\,\left(V_a^{(i)}-\mu_a\right)\,,\label{eq:alg3_a}\\
 L_a^{(\sim i)} &\equiv \frac{R^{(\sim i)}}{N-1}\left[(N-2)\,C_a + N\,\left(G^{(i)}_a-\smean\left[G_a\right]\right)\right]\,,\label{eq:alg3_b}\\
 W_\text{res}^{(i)} &\equiv W^{(i)} - \sum_{a,b=1}^{N_\mathrm{cv}}K^+_{ab}\,C_a\,\left(V_b^{(i)}-\mu_b\right)\,,\label{eq:alg3_c}\\
 S_a^{(\neg i)} &\equiv \scovsans^{(\neg i)}\left[X, G_a\right]\,,\label{eq:alg3_d}\\
 Z^{(\sim i)}_\mathrm{cv} &\equiv \left[\frac{N\,W^{(i)} - W_\text{res}^{(i)}}{N-1}\right]\,R^{(\sim i)} - \sum_{a,b=1}^{N_\mathrm{cv}} K^+_{ab}\,S^{(\neg i)}_a\,\left(V_b^{(i)}-\mu_b\right)\,,\\
 Y^{(\sim i)}_\mathrm{cv} &\equiv W^{(i)}\,X^{(i)} - \smeansans^{(\neg i)}\left[X\right]\,\left(W^{(i)}-W_\text{res}^{(i)}\right) - \sum_{a,b=1}^{N_\mathrm{cv}} K^+_{ab}\,S^{(\neg i)}_a\,\left(V_b^{(i)}-\mu_b\right) \label{eq:alg3_f}\\
 &= Z^{(\sim i)}_\mathrm{cv} + W_\text{res}^{(i)}\,\smean\left[X\right]\,,\\
 \Test{cv} &\equiv \smean\left[Y_\mathrm{cv}\right] = \smean\left[Z_\mathrm{cv}\right] + \smean\left[W_\text{res}\right]\,\smean\left[X\right]\,, \label{eq:alg3_h}\\
 \sigmaest{cv}^2 &\equiv \frac{1}{N}\,\svar[Y_\mathrm{cv}] + \frac{1}{(N-2)(N-3)}\sum_{a,b=1}^{N_\mathrm{cv}}K^+_{ab}\,\scov\left[L_a, L_b\right]\,.\label{eq:alg3_i}
 %H_{ab}^{(i)} &\equiv W^{(i)}\,\left(V_a^{(i)}-m_a\right)\,\left(V_b^{(i)}-m_b\right)\,,\\
 % \sigmaest{cv,high-acc\_low-prec}^2 &\equiv \frac{1}{N}\,\svar[Y_\mathrm{cv}] + \frac{1}{(N-1)^2}\sum_{a,b,c,d=1}^{N_\mathrm{cv}}K^+_{ab}\,K^+_{cd}\,\scov\left[X, H_{bc}\right]\,\scov\left[X, H_{ad}\right]\\
 % &\qquad\qquad + .
\end{align}
% \end{empheq}
\end{subequations}
\caption{\small Formulas to compute $\widehat{T}_\text{cv}$, the control-variates-based Monte Carlo estimator for $T$, and $\widehat{\sigma}_\text{cv}^2$, and estimate for the variance of $\widehat{T}_\text{cv}$.}
\label{alg:cv_estimator}
\end{algorithm}

\subsection{Properties of the Estimators}

Several properties of the three estimators are compiled in \tref{tab:estimator_properties}, with the proofs in \sref{sec:methodology}, \aref{appendix:uncertainty_estimates}, and/or \aref{appendix:estimator_properties}. To highlight some important ones, all three estimators for $T$ are unbiased, and the corresponding error-squared estimates $\sigmaest{}^2$ are non-negative. $\sigmaest{basic}^2$ is an unbiased estimator for $\var\left[\Test{basic}\right]$. On the other hand, $\sigmaest{centered}^2$ and $\sigmaest{cv}^2$ are slight overestimates of the variances of the corresponding $\Test{}$-s---they have non-negative biases of order $\mathcal{O}(N^{-3})$ and $\mathcal{O}(N^{-2})$, respectively. In all three cases, $\sigmaest{}$ is a slight overestimate for the standard deviation of $\Test{}$, with a non-negative bias of order $\mathcal{O}(N^{-1})$. Furthermore, in all three cases the relative Mean Squared Error (MSE) of $\sigmaest{}$ and $\sigmaest{}^2$, with respect to the variance and standard deviation of $\Test{}$, respectively is $\mathcal{O}(N^{-1})$. Here, the relative MSE is defined as follows:
\begin{align}
 \text{Relative MSE of } \hat{\varphi} \text{ with respect to } \varphi &\equiv \frac{\E\left[\left(\hat{\varphi}-\varphi\right)^2\right]}{\varphi^2} = \frac{\var\left[\hat{\varphi}\right] + \mathrm{Bias}^2\left[\hat{\varphi}\right]}{\varphi^2}\,,\label{eq:rmse_def}\\
 \qquad\text{where }\qquad\mathrm{Bias}\left[\hat{\varphi}\right] &\equiv \E\left[\hat{\varphi}\right] - \varphi\,.
 % \\
 % \text{Relative bias of } \hat{\theta} \text{ with respect to } \theta &\equiv \E\left[\frac{\hat{\theta}-\theta}{\theta}\right]\,.
\end{align}
\begin{table}[ht]
 \centering
 \caption{Properties of the different estimators defined in \sref{sec:formulas}.}
 \label{tab:estimator_properties}
 \small \setlength{\tabcolsep}{3pt}
 \begin{tabular}{M{.42\textwidth} M{.14\textwidth} M{.14\textwidth} M{.22\textwidth}}
  \toprule
  & \textbf{Basic Estimator} & \textbf{Centered Estimator} & \textbf{Control-Variate-Based Estimator}\\
  \midrule[.8pt]
  $\Test{}$ is unbiased. i.e., $\E\left[\,\Test{}\,\right]=T$ & \cmark & \cmark & \cmark\\
  \cmidrule(lr){1-4}
  $\sigmaest{}^2$ is non-negative, i.e., $\sigmaest{}^2 \geq 0$ & \cmark & \cmark & \cmark\\
  \cmidrule(lr){1-4}
  Permutation invariance: $\Test{}$ and $\sigmaest{}^2$ are unchanged under permutations of the dataset & \cmark & \cmark & \cmark\\
  \cmidrule(lr){1-4}
  Scale invariance: Under $X^{(i)}\mapsto a\,X^{(i)}$,\linebreak $\Test{}\mapsto a\,\Test{}$ and $\sigmaest{}^2\mapsto a^2\,\sigmaest{}^2$ & \cmark & \cmark & \cmark\\
  \cmidrule(lr){1-4}
  Shift invariance: Under $X^{(i)}\mapsto X^{(i)}+\delta$,\linebreak $\Test{}\mapsto \Test{}+\mu_W\delta$ and $\sigmaest{}^2\mapsto \sigmaest{}^2$ & \xmark & \cmark & {True if $(W-\mu_W)$ lies in the span of $(V_a-\mu_a)$-s}\\
  \cmidrule(lr){1-4}
  $\Test{}$ and $\sigmaest{}^2$ are invariant under invertible linear transformations of $V_a$ (with appropriate changes to $\mu_a$, $C_a$, and $K^+_{ab}$) & Not applicable & Not applicable & \cmark\\
  \cmidrule(lr){1-4}
  \multicolumn{4}{c}{\textbf{Scaling Properties}}\\
  \cmidrule(lr){1-4}
  Variance of $\Test{}$:~~ $\var\left[\,\Test{}\,\right]$ & $\mathcal{O}(N^{-1})$ & $\mathcal{O}(N^{-1})$ & $\mathcal{O}(N^{-1})$\\
  \cmidrule(lr){1-4}
  % Standard deviation of $\Test{}$:~~ $\sqrt{\var\left[\,\Test{}\,\right]}$ & $\mathcal{O}(N^{-1/2})$ & $\mathcal{O}(N^{-1/2})$ & $\mathcal{O}(N^{-1/2})$\\
  % \cmidrule(lr){1-4}
  Bias of $\sigmaest{}^2$:~~ $\E\left[\sigmaest{}^2\right] - \var\left[\,\Test{}\,\right]$ & $= 0$ & $\geq 0$; $\mathcal{O}(N^{-3})$ & $\geq 0$; $\mathcal{O}(N^{-2})$\\
  \cmidrule(lr){1-4}
  Variance of $\sigmaest{}^2$:~~ $\var\left[\sigmaest{}^2\right]$ & $\mathcal{O}(N^{-3})$ & $\mathcal{O}(N^{-3})$ & $\mathcal{O}(N^{-3})$\\
  \cmidrule(lr){1-4}
  Bias of $\sigmaest{}$:~~ $\E\left[\sigmaest{}\right] - \sqrt{\var\left[\,\Test{}\,\right]}$ & $\geq 0$; $\mathcal{O}(N^{-1})$ & $\geq 0$; $\mathcal{O}(N^{-1})$ & $\geq 0$; $\mathcal{O}(N^{-1})$\\
  \cmidrule(lr){1-4}
  Variance of $\sigmaest{}$:~~ $\var\left[\sigmaest{}\right]$ & $\mathcal{O}(N^{-2})$ & $\mathcal{O}(N^{-2})$ & $\mathcal{O}(N^{-2})$\\
  \cmidrule(lr){1-4}
  % MSE of variance estimate, $\var\left[\sigmaest{}^2\right]$ & $\mathcal{O}(N^{-3})$ & $\mathcal{O}(N^{-3})$ & $\mathcal{O}(N^{-3})$\\
  % \cmidrule(lr){1-4}
  Relative MSE of $\sigmaest{}^2$ and $\sigmaest{}$, with respect to $\var\left[\,\Test{}\,\right]$ and $\sqrt{\var\left[\,\Test{}\,\right]}$, respectively & $\mathcal{O}(N^{-1})$ & $\mathcal{O}(N^{-1})$ & $\mathcal{O}(N^{-1})$ \\
  % \cmidrule(lr){1-4}
  % Relative bias of $\sigmaest{}^2$ with respect to $\var\left[\,\Test{}\,\right]$ & $= 0$ & $\geq 0$; $\mathcal{O}(N^{-2})$ & $\geq 0$; $\mathcal{O}(N^{-1}N_\mathrm{cv}^{???})$\\
  % \cmidrule(lr){1-4}
  % Relative bias of std. dev. estimate, $\frac{\E\left[\sigmaest{}\right] - \sqrt{\var\left[\,\Test{}\,\right]\,}}{\sqrt{\var\left[\,\Test{}\,\right]\,}}$ & $\geq 0$; $\mathcal{O}(N^{-1})$ & $\geq 0$; $\mathcal{O}(N^{-1})$ & $\geq 0$; $\mathcal{O}(N^{-1})$\\
  \bottomrule
 \end{tabular}
\end{table}

\subsection{Computational Complexities}
The computational complexities of the various operations involved in performing a control variates based estimation are listed in \tref{tab:complexity}. In the first two rows, the weight $W$ and the controls $V_a$ are assumed to be factorizable, like $V_\mathrm{fac}$ in \eqref{eq:v_fac}. Note that $\Test{cv}$ and $\sigmaest{cv}$ are invariant under invertible linear transformations of the controls, i.e., transformations which do not change the span of $(V_a-m_a)$-s. One particular such transformation is the diagonalization the covariance matrix (third row of \tref{tab:complexity}), which would make the subsequent operations (fourth and fifth rows) faster. This would be worthwhile, especially if the same set of controls is to be used for multiple estimations. In the simulation study described in \sref{sec:experiments}, the values of $M$ and $N_{\Sigma K}$ are much larger than $N$ and $N_\mathrm{cv}$. In such situations, the computational bottlenecks (aside from running quantum circuits to collect $X^{(i)}$-values) in performing the control-variates based estimation would be the operations in the first two rows of \tref{tab:complexity}.
\begin{table}[ht]
 \centering
 \caption{The computational complexity of the different steps involved in performing a control-variates-based estimation. $N$ is the number of datapoints, $N_\mathrm{cv}$ is the number of control variates, and $N_{\Sigma K}$ is as defined in \eqref{eq:Npisumdef}.}
 \label{tab:complexity}
 \small
 \begin{tabular}{M{.4\textwidth} M{.45\textwidth}}
 \toprule
 \textbf{Operation} & \textbf{Computational complexity}\\
 \midrule[.8pt]
 Computing the values of  $\mu_W$, $\mu_a$ $C_a$, $K_{ab}$ for factorizable weight and control variates & $\mathcal{O}(N_{\Sigma K}\,N_\mathrm{cv}^2)$ if $\mathbf{K}$ is not inherently diagonal\linebreak $\mathcal{O}(N_{\Sigma K}\,N_\mathrm{cv})$ if $\mathbf{K}$ is inherently diagonal \\
 \cmidrule(lr){1-2}
 Computing the values of $W^{(i)}$ and $V_a^{(i)}$ for factorizable weight and control variates & $\mathcal{O}(N_{\Sigma K}\,N_\mathrm{cv}\,N)$\\
 \cmidrule(lr){1-2}
 (optional) Diagonalizing $\mathbf{K}$ and performing the corresponding rotations on $\mu_a$, $C_a$, $V_a^{(i)}$ & $\mathcal{O}(N_\mathrm{cv}^3 + N^2_\mathrm{cv}\,N)$\\
 \cmidrule(lr){1-2}
 Computing $\mathbf{K}^+$ from $\mathbf{K}$ & $\mathcal{O}(N_\mathrm{cv})$ if $\mathbf{K}$ is diagonal\linebreak $\mathcal{O}(N_\mathrm{cv}^3)$ if $\mathbf{K}$ is not diagonal\\
 \cmidrule(lr){1-2}
 Computing $\Test{cv}$ and $\sigmaest{cv}$ using \alref{alg:cv_estimator} & $\mathcal{O}(N_\mathrm{cv}\,N)$ if $\mathbf{K}$ is diagonal\linebreak $\mathcal{O}(N_\mathrm{cv}^2\,N)$ if $\mathbf{K}$ is not diagonal\\
 \bottomrule
 \end{tabular}
\end{table}

\section{Experiments and Results} \label{sec:experiments}
In this section we will demonstrate the application of our technique for the reduction of sampling overhead in PEC, using simulation experiments that are designed to imitate a realistic PEC setup studied in Ref.~\cite{vandenBerg2023}.

\subsection{Simulation Experiments Description}
\begin{figure}
 \centering
 \includegraphics[width=.8\textwidth]{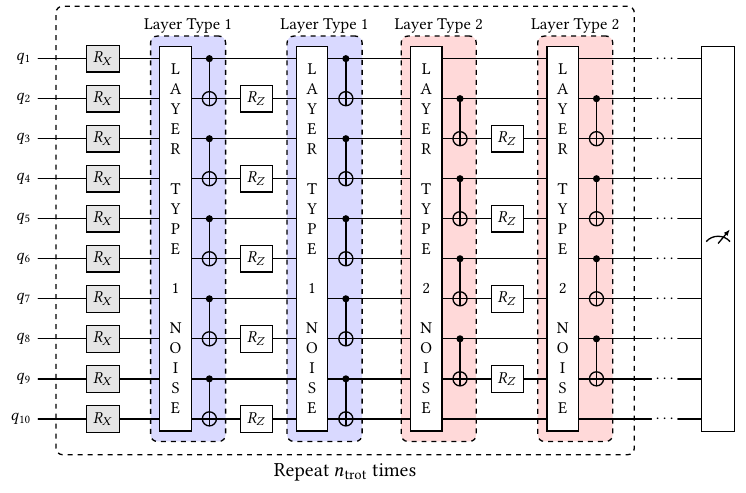}
 \caption{Diagram depicting the 10-qubit circuits, which simulate the Trotterized time-evolution, over $n_\mathrm{trot}$ Trotter steps, of an Ising model with ten sites. The 4-qubit circuits considered in this study are similar to the 10-qubits circuits, except that only the first four qubits in the diagram are kept.}
 \label{fig:circuit}
\end{figure}
\subsubsection{Target Quantum Circuits} Similar to Ref.~\cite{vandenBerg2023}, we consider circuits that simulate the Trotterized time-evolution of the one-dimensional transverse-field Ising model, for spin-chains with 4 and 10 sites. The circuits for the case of ten Ising-sites is depicted in \fref{fig:circuit}, where $R_X$ and $R_Z$ are one-qubit rotation gates with matrix representations given by
\begin{alignat}{2}
 R_X &\equiv \begin{bmatrix}
  \cos(\theta_x/2) & -i\sin(\theta_x/2)\\
  -i\sin(\theta_x/2) & \cos(\theta_x/2)
 \end{bmatrix}\,,\qquad &&\text{where } \theta_x = 2\,h\,\delta_t\,,\\[1em]
 R_Z &\equiv \begin{bmatrix}
  \exp(-i\,\theta_z/2) & 0\\
  0 & \exp(i\,\theta_z/2)
 \end{bmatrix}\,,\qquad &&\text{where } \theta_z = -2\,J\,\delta_t\,.
\end{alignat}
Here $h$ and $J$, both real-valued, are the Ising model parameters denoting the transverse magnetic field and the exchange coupling between neighboring sites, respectively. $\delta_t$ is a real-valued parameter denoting the Trotter step-size and $n_\mathrm{trot}$ is the number of Trotter steps simulated by the circuit. The circuits for the four Ising-sites case are identical to those for the ten Ising-sites case, except that only the first four qubits (and the gates which act exclusively on them) are kept. We use the following values for $h$, $J$, and $\delta_t$:
\begin{alignat}{4}
 &\text{In the 4-qubit circuits: }\quad &&h=1\,,\quad &&J=0.15\,,\quad &&\delta_t=0.5\,,\\
 &\text{In the 10-qubit circuits: }\quad &&h=1\,,\quad &&J=-0.5236\,,\quad &&\delta_t=0.5\,.
\end{alignat}
For the 4-qubit case, we run circuits with $n_\mathrm{trot}$ values ranging from 1 to 15. For the 10-qubit case, $n_\mathrm{trot}$ ranges from 1 to 7. We will refer to (the noiseless versions of) the $15+7=22$ circuits described here as ``target circuits'' in the rest of this paper.

\subsubsection{Noise Model}
We use the Pauli--Lindblad noise model studied in Ref.~\cite{vandenBerg2023} in our simulations. All one-qubit gates are taken to be noiseless. Each layer of two-qubit gates in the circuits has an associated stochastic noise, which is modeled as a Pauli noise-channel that acts immediately \emph{before} the noiseless version of that layer, as shown in \fref{fig:circuit}. The Pauli noise channel for a given layer is modeled as arising from a sparse set of local (non-identity) Pauli interactions, each of which acts either on only one qubit or on two neighboring qubits in the circuit (with qubits $q_i$ and $q_{i+1}$ taken to be neighbors for all $1\leq i<\text{number of qubits}$). The total number of such Pauli interactions, say $N_\mathrm{Paulis}$, is given by.
\begin{alignat}{2}
 N_\mathrm{Paulis} &= (4\times 3) + (3\times 9) = 39\,,\quad&& \text{for 4-qubit-circuits}\,,\\
 N_\mathrm{Paulis} &= (10\times 3) + (9\times 9) = 111\,,\quad&& \text{for 10-qubit-circuits}\,.
\end{alignat}
Note that in all the target circuits with a given number of qubits (say 10), there are only two types of layers containing two-qubit gates; these are labeled as ``Layer Type 1'' and ``Layer Type 2'' in \fref{fig:circuit}.
% Note that in all the target circuits with a given number of qubits (either 4 or 10), each two-qubit-gates-layer is of one of only two types, labeled ``Layer Type 1'' and ``Layer Type 2'' in \fref{fig:circuit}.
All layers of type 1 are assumed to have the same Pauli noise channel, and likewise for layers of type 2.

Concretely, the operation $U_\text{\sc t}$, which corresponds to a two-qubit-gates-layer of type $\text{T}\in\{1, 2\}$, is given by $U^\mathrm{ideal}_\text{\sc t} \circ \Lambda_\text{\sc t}$, where $U^\mathrm{ideal}_\text{\sc t}$ is the noiseless version of $U_\text{\sc t}$ and $\Lambda_\text{\sc t}$ is a stochastic noise channel given by
\begin{align}
 \Lambda_\text{\sc t}(\rho) &\equiv \left[\Big.\Lambda_{\text{\sc t}, N_\mathrm{Paulis}}\circ\dots\circ\Lambda_{\text{\sc t},1}\right](\rho)\,,\\
 \text{where}\,,\qquad\qquad\Lambda_{\text{\sc t},i}(\rho) &\equiv \left(1-\epsilon_{\text{\sc t},i}\right)\,\rho + \epsilon_{\text{\sc t},i}\,P_i\,\rho\,P_i^\dagger\,,\qquad \text{with } 0\leq \epsilon_{\text{\sc t},i} < 0.5
\end{align}
Here $\Lambda_{\text{\sc t},i}$ is the noise-channel associated with the $i$-th local Pauli interaction (represented by the unitary operator $P_i$), and $\epsilon_{\text{\sc t},i}$ is its error rate. The noise-model is completely specified by the values of the parameters $\epsilon_{\text{\sc t},i}$, for all relevant $P_i$-s and both values of $\text{T}$.

We use the same values for the parameters as in Ref.~\cite{vandenBerg2023}
% (listed in \tref{tab:4_qubits_noise_model_params} and \tref{tab:10_qubits_noise_model_params} for completeness).
(listed in \aref{appendix:noise_model_params} for completeness).
In that work, the parameter-values were estimated from hardware experiments on IBM's quantum processor named {\sc ibm\_hanoi} \cite{ibmhanoi} using a novel learning protocol. Furthermore, performing PEC on {\sc ibm\_hanoi} with these noise-parameter-values was demonstrated to provide results which accurately matched noiseless-simulations. This evidences that the noise-model a) was sufficiently accurate for the purposes of Ref.~\cite{vandenBerg2023},\footnote{In Ref.~\cite{vandenBerg2023}, the hardware noise is adapted using the Pauli twirling technique, so that the noise operator of each layer is a Pauli channel. This is essential for the chosen noise model to accurately describe the hardware noise.} and b) is realistic for the purposes of our study.

\subsubsection{PEC Implementation}
Our PEC implementation, described here for completeness, also mirrors Ref.~\cite{vandenBerg2023}. The inverse of $\Lambda_\text{\sc t}$ is given by
\begin{align}
 \Lambda_\text{\sc t}^{-1}(\rho) &= \left[\Lambda_{\text{\sc t},i}^{-1}\circ\dots\circ\Lambda_{\text{\sc t},N_\mathrm{paulis}}^{-1}\right](\rho)\,,\qquad\qquad\text{where}\,,\\
 \Lambda_{\text{\sc t},i}^{-1}(\rho) &= \frac{1}{1-2\,\epsilon_{\text{\sc t},i}}\left[\Big.\left(1-\epsilon_{\text{\sc t},i}\right)\,\rho - \epsilon_{\text{\sc t},i}\,P_i\,\rho\,P_i^\dagger\right]\,.
\end{align}
Note that the operations $\Lambda_{\text{\sc t},i}^{-1}$ and $\Lambda_{\text{\sc t},j}^{-1}$ commute with each other. The effect of the noise on the computation-results can be cancelled by quasi-probabilistically performing these inverse channels immediately \emph{before} performing the corresponding noisy two-qubit-gate-layers. Operationally, for every local Pauli noise-term $\Lambda_{\text{\sc t},i}$ (of every noisy layer in every Trotterization step in the circuit),
one quasi-probabilistically performs $\Lambda_{\text{\sc t},i}^{-1}$ by either 
inserting or not inserting (with a certain probability) the corresponding Pauli-gate $P_i$, immediately before the layer.
% one either inserts or does not insert (with a certain probability) the corresponding Pauli-gates, immediately before the layer.
The resulting random PEC ``mitigation circuit instance'' corresponds to the operation $\mathcal{E}_{(k_1,\dots,k_M)}$ in \eqref{eq:qpd}.

Under the notation used in this paper, $M$ is the total number of such binary decisions to be made to generate a single mitigation circuit instance.
\begin{align}
 M &= 4\times N_\mathrm{Paulis} \times n_\mathrm{trot}\,,\\
 K_m &= 2\,,\qquad\qquad\qquad\qquad\forall m\in\{1,\dots,M\}\,.
\end{align}
For each value of $m$, let $\epsilon(m)$ represent the error rate $\epsilon_{\text{\sc t},i}$ of the corresponding noise-term, and let $k_m$ being $1$ and $2$ correspond to the choice of not inserting and inserting the corresponding Pauli, respectively. Under our notation
\begin{align}
 q_m(1) &= \frac{1-\epsilon(m)}{1-2\,\epsilon(m)}\,,\qquad\qquad q_m(2) = \frac{-\,\epsilon(m)}{1-2\,\epsilon(m)}\,.
\end{align}
Following the standard practice, we chose the PEC sampling probabilities $p_m$ to be proportional to the absolute value of $q_m$:
\begin{alignat}{2}
 p_m(1) &= 1-\epsilon(m)\,,\qquad &&p_m(2) = \epsilon(m)\,,\\
 \Longrightarrow w_m(1) &= +\gamma_m\,,\qquad &&w_m(2) = -\gamma_m\,,\qquad\qquad\text{where}\quad\gamma_m = \frac{1}{1-2\,\epsilon(m)}\,.
\end{alignat}
For the 4-qubit circuits, the value of $\gamma=\prod_m \gamma_m$, which is a gauge of the sampling overhead, ranges from approximately $1.13$ for $n_\mathrm{trot}=1$ to approximately $6.48$ for $n_\mathrm{trot}=15$. For the 10-qubit circuits, $\gamma$ ranges for approximately $1.66$ for $n_\mathrm{trot}=1$ to approximately $34.95$ for $n_\mathrm{trot}=7$.

\subsubsection{Simulation Details}
All the quantum circuit simulations in this study were performed using Qiskit \cite{Qiskit}. For each target circuit considered in this paper, a total of 200 mitigation circuit instances were sampled. Each of these mitigation circuit instances was simulated with noise to get $1024$ Pauli-Y measurements of all the qubits and, independently, another $1024$ Pauli-Z measurements of all the qubits. The stochastic noise considered in this paper was incorporated into the simulations using the trajectory simulation technique \cite{2111.02396}. Stated explicitly, to simulate each shot of a given mitigation circuit instance, a different noise-trajectory-instance was randomly sampled from the appropriate distribution and incorporated into the circuit.\footnote{The Pauli gates inserted for the purposes of PEC-based mitigation and/or noise-simulation were grouped and simplified where appropriate.}

In addition to these, for each target circuit, we also performed a) noisy simulations without PEC ($1024$ shots each in the Pauli-Y and Pauli-Z measurement bases), and b) noiseless simulations ($1024^2$ shots each in the Pauli-Y and Pauli-Z measurement bases) for validation purposes. Overall, the generation of the data used in this study involved simulating
\begin{itemize}
 \item $15\times(2 + (1 + 200)\times 2\times 1024) = 6{,}174{,}750$ four-qubit circuit instances, and
 \item $7\times(2 + (1 + 200)\times 2\times 1024) = 2{,}881{,}550$ ten-qubit circuit instances
\end{itemize}
of varying depths.\footnote{Some of these circuit instances may have been identical to each other. Each shot in a noisy simulation is counted as a circuit instance to be simulated.} In addition to the results of the circuits, we also recorded the value of the indices $(k_1,\dots,k_M)$ for each mitigation circuit instance, since this is needed for performing our control variates technique.

\subsection{Estimation Tasks}
Let $\mathcal{M}_{q_i,y}\in\{-1, +1\}$ represent a Pauli-Y measurement outcome (in a single shot) for the $i$-th qubit. Likewise, let $\mathcal{M}_{q_i,z}$ represent a Pauli-Z measurement outcome. For a circuit with $Q$ qubits, let $O_{\kappa,y}$ and $O_{\kappa,z}$ represent the average of all possible weight-$\kappa$ Pauli-Y and Pauli-Z observables, respectively, that one can construct using the $Q$ qubits.
\begin{align}
 O_{\kappa,y} &\equiv \frac{\kappa!\,(Q-\kappa)!}{Q!}~~\sum_{\substack{1\leq q_1<\ldots\\\ldots<q_\kappa\leq Q}}
 \left[\prod_{i=1}^\kappa \mathcal{M}_{q_i,y}\right]\,,\qquad\qquad\forall\kappa\in\{1,\dots,Q\}\,,\\
 O_{\kappa,z} &\equiv \frac{\kappa!\,(Q-\kappa)!}{Q!}~~\sum_{\substack{1\leq q_1<\ldots\\\ldots<q_\kappa\leq Q}} \left[\prod_{i=1}^\kappa \mathcal{M}_{q_i,z}\right]\,,\qquad\qquad\forall\kappa\in\{1,\dots,Q\}\,.
\end{align}
Furthermore, let $O_{2\text{-nearest},y}$ and $O_{2\text{-nearest},z}$ represent the average of all weight-2 Pauli-Y and Pauli-Z observables, respectively, which involve neighboring qubits in the circuit.
\begin{align}
 O_{2\text{-nearest},y} &\equiv \frac{1}{Q-1}\sum_{q=1}^{Q-1}~\mathcal{M}_{q,y}\,\mathcal{M}_{q+1,y}\,,\\
 O_{2\text{-nearest},z} &\equiv \frac{1}{Q-1}\sum_{q=1}^{Q-1}~\mathcal{M}_{q,z}\,\mathcal{M}_{q+1,z}
\end{align}
For each target circuit, we will estimate the noiseless-expectations of the observables $O_{\kappa,y}$ and $O_{\kappa,z}$ (with $\kappa$ running from $1$ to $Q$), as well as $O_{2\text{-nearest},y}$ and $O_{2\text{-nearest},z}$ using the different formulas in \sref{sec:formulas}. The quantity $X^{(i)}$ in these formulas is the shot-average of the relevant observable (over 1024 shots), for the $i$-th mitigation circuit instance (out of 200).

We will refer to the estimation of a given observable for a given target circuit as an ``estimation task'' in the rest of this paper; a total of $(15\times 2\times 5) + (7\times 2\times 11) = 304$ estimation tasks are considered in this study. Note that for a given target circuit, the estimations for all the Pauli-Y-based-observables are performed using the same measurement-results, and likewise for all the Pauli-Z-based-observables. Furthermore, the exact same 200 PEC circuit-instances are used for the Pauli-Y and the Pauli-Z measurements (the noise trajectories, however, are sampled independently). Consequently, the results of different estimation tasks corresponding to the same target circuit are mutually dependent.
% As a result, the estimations of different observables for the same target circuit are dependent on each other.

\subsection{Choice of Control Variates} \label{sec:experiments_cvs}
We consider the following five sets of control variates to perform the control-variates-based estimations. CV set nos. 2, 3, and 5 contain factorizable control variates $V_a$ of the form
\begin{subequations}\label{eq:Va_form}
\begin{align}
 V_a(k_1,\dots,k_M) &\equiv \prod_{m=1}^M \frac{\tilde{v}_{a,m}(k_m)}{\mathtt{norm\_factor}_{a,m}}\,,\qquad\text{where}\\
 \mathtt{norm\_factor}_{a,m} &\equiv \sqrt{\sum_{k_m=1}^M p_m(k_m)\,\left[\big.\tilde{v}_{a,m}(k_m)\right]^2~}
\end{align}
\end{subequations}
here, the scaling factors $\mathtt{norm\_factor}_{a,m}$ are included simply to impose the property $\E\left[V_a^2\right]=1$; in principle (with arbitrary-precision arithmetic), the scaling factor will not affect the results of the estimations performed using these controls.

\subsubsection{CV Set 1}
This set only contains only one control $V_1\equiv \sign(W)$, which is proportional to $W$.

\subsubsection{CV Set 2}
This set contains five controls of the form in \eqref{eq:Va_form}, with values of $\tilde{v}_{a,m}$ given by
\begin{align}
 \tilde{v}_{a,m}(1) = \theta_a+1\,,\qquad&\qquad \tilde{v}_{a,m}(2) = \theta_a-1\,,\\
 \text{with}\quad\theta_1 = -1.5\,,\quad \theta_2 = -0.75\,,\quad &\theta_3 = 0\,,\quad \theta_4=0.75\,,\quad \theta_5=1.5\,.
\end{align}
Note that $V_3$ is proportional to $W$. The other controls can be thought of as modifications of the weight. The effectiveness of $W$ as a control, as discussed in \sref{sec:weight_as_control}, motivates this construction. The values of $\theta_a$ were not optimized in any meaningful way.\footnote{The parameters were just chosen to be equidistant, include $0$ and exclude $\pm 1$ (to prevent the controls from becoming 0 with high probability).}

% \begin{align}
%  \phi_1 = -4\,,\quad \phi_2 = 8\,,\quad \phi_3 = 2\,,\quad \phi_4=8/7\,,\quad \phi_5=4/5\,.
% \end{align}

\subsubsection{CV Set 3}
This set contains five controls of the form in \eqref{eq:Va_form}, with values of $\tilde{v}_{a,m}$ given by
\begin{align}
 \tilde{v}_{a,m}(1) = 1\,,\qquad&\qquad \tilde{v}_{a,m}(2) = \phi_a-1\,,\\
 \text{with}\quad\phi_1 = -3\,,\quad \phi_2 = -1.5\,,\quad &\phi_3 = 0\,,\quad \phi_4=1.5\,,\quad \phi_5=3\,.
\end{align}
As in CV set 2, $V_3$ is proportional to $W$ and the other controls can be thought of as modifications of the weight.\footnote{One can reproduce CV set 3 with different choices for $\theta_a$-s in CV set 2, and vice versa.}

% \begin{align}
%  \theta_1 = -3\,,\quad \theta_2 = 3\,,\quad \theta_3 = 0\,,\quad \theta_4=-3/7\,,\quad \theta_5=-3/5\,.
% \end{align}

\subsubsection{CV Set 4}
In CV set 4, we create controls which are products of subsets $w_m$-s, grouped based on the qubits on which the different noise terms act. For a circuit with $Q$-qubits, this set contains $2Q$ controls given by
\begin{alignat}{2}
 V_q &~~\equiv~~ \prod_{\substack{m \in \text{Noise terms which }\\\quad\text{only affect qubit }q}} \sign\left(\big.w_m\right)\,,\qquad\qquad&&\forall q\in\{1,\dots, Q\}\,,\\
 V_{Q+q} &~~\equiv~~ \prod_{\substack{m \in \text{Noise terms which affect }\\\quad\text{qubits }q\text{ and }q+1, \text{ both}}} \sign\left(\big.w_m\right)\,,\qquad\qquad&&\forall q\in\{1,\dots, Q-1\}\,,\\
 V_{2Q} &~~\equiv~~ \prod_{m\in\text{All noise terms}} \sign\left(\big.w_m\right) \qquad\qquad = \sign\left(\big.W\right)\,,
\end{alignat}
with $V_{2Q}$ being proportional to $W$.

\subsubsection{CV Set 5}
This set contains five controls of the form in \eqref{eq:Va_form}, with each parameter $\tilde{v}_{a,m}(k_m)$ sampled independently from the standard normal distribution. We use the same values of $\tilde{v}_{a,m}(k_m)$ for all estimations performed for a given target circuit. However, the parameters are sampled independently for different target circuits.\footnote{Note that each of the five cv sets corresponds to different controls for different target circuits, despite the cv set id being the same.} This set is not expected to be effective for variance reduction purposes. It is included only to demonstrate that the success of our technique depends on the ability to construct effective control variates.

\subsection{Note on numerical stability}

The value of $M$ gets sufficiently large in our experiments that one can run into underflow or overflow issues with finite precision arithmetic, when multiplying $M$ numbers to compute the controls as in
%\eqref{eq:fac_expectation} or
\eqref{eq:Va_form}.\footnote{If $Y \equiv y_1\times\dots\times y_M$, with the $y_m$-s being independent random variables, the dynamic range (DR) (i.e., ratio between the largest and smallest values) of $|Y|$ equals the product of DRs of $|y_m|$-s. In this way, the DR of $|Y|$ grows exponentially in $M$.
%Similarly, $\E[|Y|^2]/\E^2[|Y|]$, another measure of the variability of $|Y|$, is a product of the corresponding values for 
Similarly, the expected value of $Y^2$ (or $|Y|$) is the product of the expected values of $y_m^2$-s (or $|y_m|$-s); this can also grow exponentially if left unchecked. The scaling factors in \eqref{eq:Va_form} solve the latter issue by ensuring that $\E[V_a^2] = 1$, but this does not solve the DR issue. Note that the DR issue is not encountered when multiplying $w_m$-s of the form $\pm \gamma_m$ to compute $W$, since the DR of each $|w_m|$ is 1 in this case.} To mitigate this problem, we use log-domain computations where necessary; this involves representing each real-valued quantity $r$ with the pair $(\sign(r), \log(|r|))$, and performing a) multiplications/divisions using additions/subtractions of the logarithms and b) summations using a numerically stable implementation of the log-sum-exp operation \cite{2020SciPy-NMeth}. No underflow/overflow issues were encountered with this measure in place.

\subsection{Results} \label{sec:experiment_results}

For each of the $304$ estimation tasks, we estimate $T$ (the expected value of the observable under a noiseless computation) using seven different estimators---the basic estimator, the centered estimator, and the cv-based estimator, using each of the five cv sets in \sref{sec:experiments_cvs}. In this section we will refer to the seven estimators as estimation ``methods''.

The estimation results in the 4-qubits case, for the observables $O_{1,y}$ and $O_{1,z}$ are shown in \fref{fig:spiral}. The middle and right panels show the results for estimations performed using the basic estimator and cv-based estimator with cv set 1, respectively (solid lines). Note that both estimations are performed with the exact same data. The error bars depict $\pm \sigmaest{method}$ for the respective estimation method. The left panel of \fref{fig:spiral} shows the results from the noisy simulations without PEC (solid line). Each panel also has the results from noiseless simulations for reference (dotted lines). In each case, the results for sequential values of $n_\mathrm{trot}$ from 1 to 15 are connected to form the spirals shown in the figures. The error bars shown with the ``noisy no PEC'' results and noiseless results correspond to $\sigmaest{no-pec}$, whose square is given by
\begin{align}
 \sigmaest{no-pec}^2 \equiv \frac{\text{sample variance of the observable outcomes across the different shots}}{\text{number of shots}}\,. \label{eq:sigmasq_nopec}
\end{align}
\begin{figure}
 \centering
 \includegraphics[width=\textwidth]{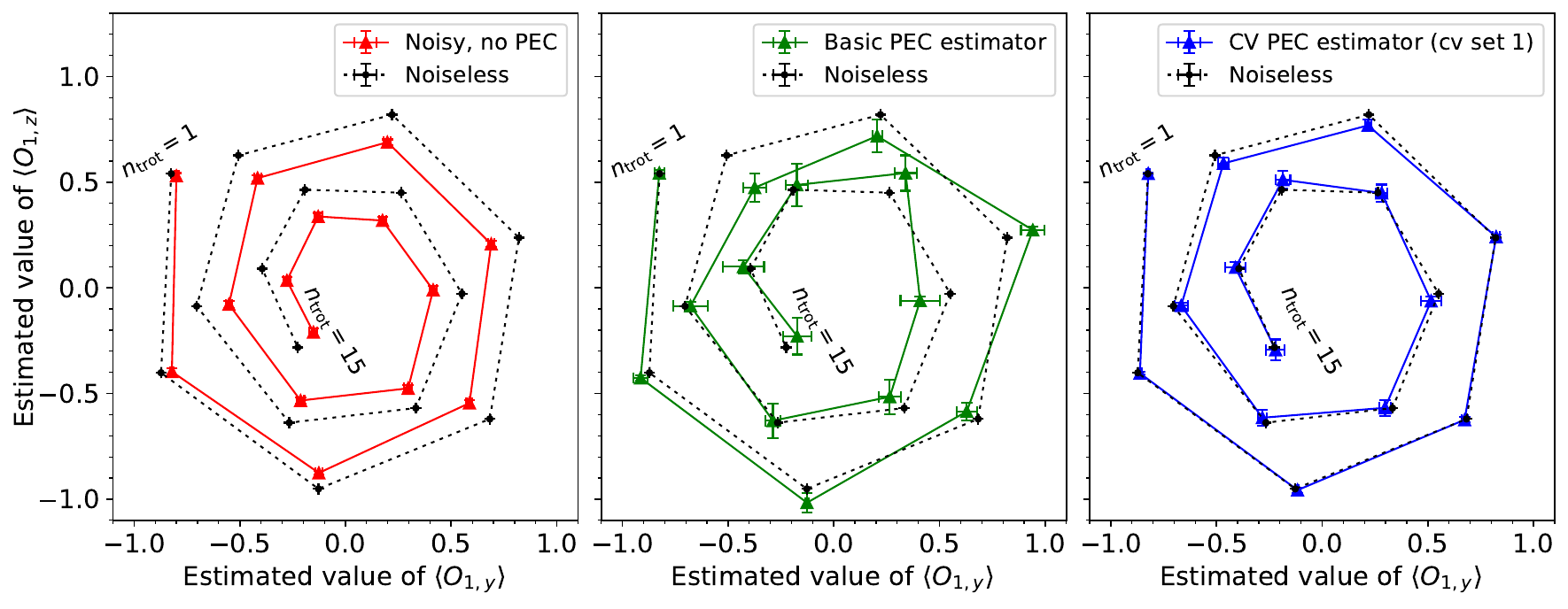}
 \caption{PEC estimation results for the observables $O_{1,y}$ and $O_{1,z}$ for the 4-qubits circuits. The middle and right panels show the results for estimations performed using the basic estimator and cv-based estimator with cv set 1. The left panel shows the results from noisy simulations without PEC. The results for sequential values of $n_\mathrm{trot}$ from 1 (outermost) and 15 (innermost) are conected to form spirals.}
 \label{fig:spiral}
\end{figure}

It can be seen that the noisy results in the left panel of \fref{fig:spiral} are inconsistent with the corresponding noiseless results. On the other hand, the PEC estimates in the middle and right panels are consistent with the noiseless results, within statistical fluctuations of magnitudes comparable to the corresponding error bars. This is not surprising since the noise model used in the PEC estimations exactly matches the noise model used in the simulation of the circuits. The error bars, in the right panel, for the cv-based estimator are substantially smaller than the corresponding error bars, in the middle panel, for the basic estimator. This demonstrates a variance reduction and, equivalently, a sampling overhead reduction. Next we will quantify this variance reduction across the various estimations performed in this study.

One way to compare the precisions of the different estimators is to perform each estimation task multiple times with each method using different independent datasets, in order to estimate the corresponding precision.
%repeat each estimation task (with a fixed target circuit, observable, and estimation method) multiple times to estimate the corresponding precision.
However, this is computationally too expensive for the purposes of this study, so we use the values of $\sigmaest{method}$ from the estimation formulas as a precision estimate.\footnote{Note that $\sigmaest{method}$ itself is a random variable with statistical fluctuations. With 200 PEC mitigation circuits the standard deviation of $\sigmaest{method}$ will be about $5\%$ of $\sqrt{\var[\Test{method}]}$, i.e., the uncertainty on the uncertainty-estimate will be about 5\% (provided the number of controls is not too large). This precision is sufficient for the purposes of this paper.} We will first validate that our formulas provide reasonably unbiased precision estimates. For every single estimation task, we define a Studentized residual as follows:
\begin{align}
 \text{Studentized residual} &\equiv \frac{\Test{method} - \Test{noiseless}}{~~~~\sqrt{\sigmaest{method}^2 + \sigmaest{noiseless}^2}~~~~} \label{eq:studentized_residual}
\end{align}
where $\Test{noiseless}$ is the estimate of $T$ from noiseless simulations and $\sigmaest{noiseless}^2$ is computed using the formula for $\sigmaest{no-pec}^2$ in \eqref{eq:sigmasq_nopec}. \Fref{fig:residual_cdf} depicts the empirical cumulative distribution function (CDF) of the absolute value of the Studentized residual under each estimation method. For each method, empirical CDF is computed from the 304 estimations performed using that method. For reference, the CDF of the absolute value of a standard normal distributed random variable is also shown in \fref{fig:residual_cdf}. The similarity of the empirical CDFs to the reference CDF validates using $\sigmaest{method}$ as an estimate of precision.\footnote{Some unimportant caveats here are that a) the estimations of different observables for the same target circuit are dependent on each other, and b) the Studentized residual is expected to only approximately follow the standard normal distribution.}
\begin{figure}
 \centering
 \includegraphics[width=.5\textwidth]{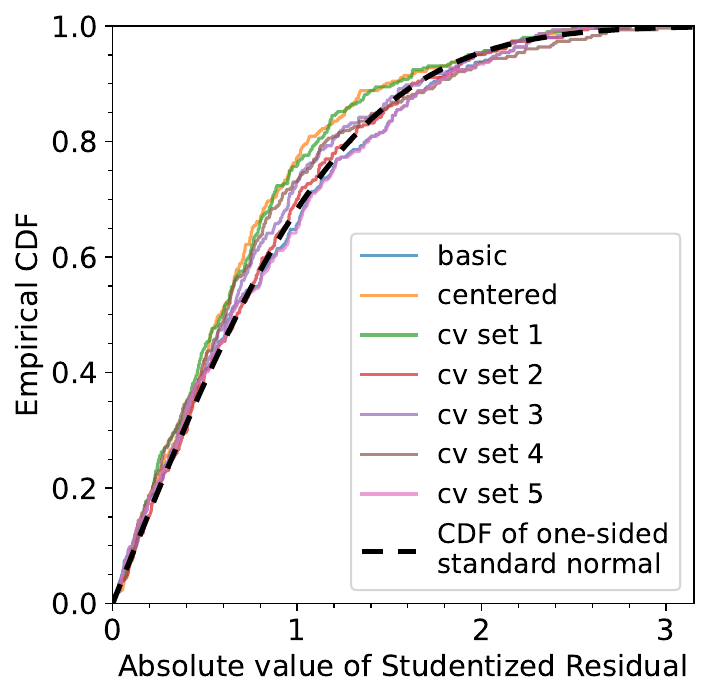}
 \caption{The empirical cumulative distribution function (CDF) of the absolute value of the Studentized residual defined in \eqref{eq:studentized_residual} under the different estimation methods. For each method, the empirical CDF is computed from the 304 estimations performed using that method.}
 \label{fig:residual_cdf}
\end{figure}

Noting that the variance of $\Test{basic}$ is proportional to $1/N$, one can define a ``Data Amplification Factor'' (DAF) of an estimation method
%(other than the basic estimator)
for a given estimation task as follows:
\begin{align}
 \text{Data Amplification Factor} &\equiv \frac{\sigmaest{basic}^2}{\sigmaest{method}^2}\,.
\end{align}
This is an estimate of the factor by which the number of datapoints $N$ has to be increased, in order for the basic estimator to reach the precision achieved (using only the original dataset) by a given method. The amplification factor can also be reparameterized in terms of a ``Sampling Overhead Reduction Percentage'' (SORP) as follows:
\begin{align}
 \text{Sampling Overhead Reduction Percentage} &\equiv \left(1-\frac{\sigmaest{method}^2}{\sigmaest{basic}^2}\right)\times 100\%\,.
\end{align}
The performance improvement from using a given method (when compared to the basic estimator) is roughly equivalent to reducing the value of $\gamma^2$ by the SORP. We will use these two performance metrics to evaluate the variance and sampling overhead reduction achieved by the different estimation methods.\footnote{Note that both these metrics are subject to statistical fluctuations, since they are based on data-based estimates of precision. Also, the variance under estimation methods other than the basic estimator is only approximately proportional to $1/N$; these metrics will increase (i.e., improve) slightly if the number of datapoints $N$ is increased.}

\Fref{fig:4-qubits-amp} shows the DAF and the SORP for the different estimation methods (except the basic estimator), for all the estimation tasks relevant to 4-qubit circuits. \Fref{fig:10-qubits-amp} shows the same for 10-qubit circuits---for brevity, the results for only 10 out of 22 observables are shown.
\begin{figure}
 \centering
 \includegraphics[width=\textwidth]{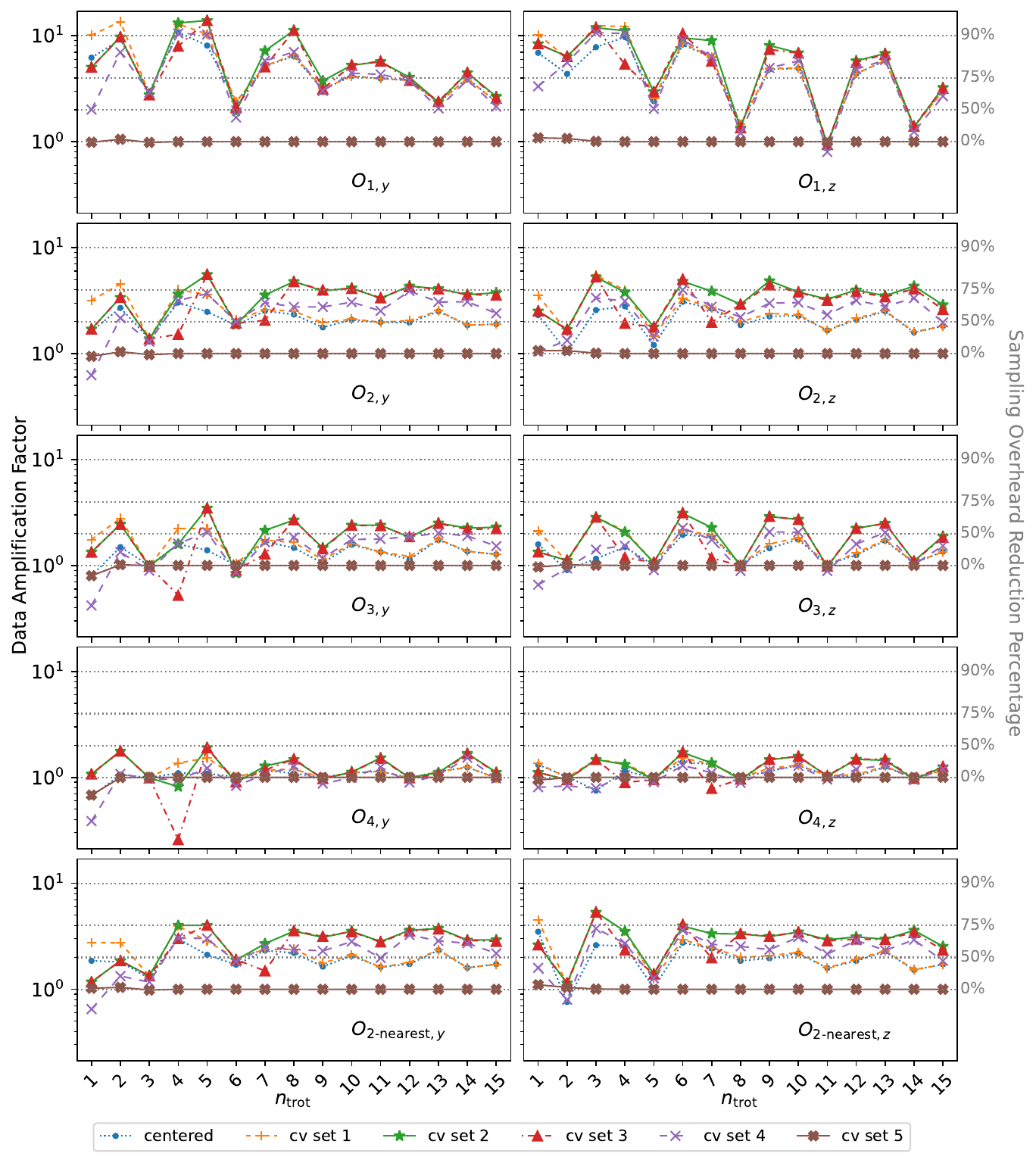}
 \caption{The Data Amplification Factor (DAF) for the different estimation methods (except the basic estimator), for all the estimation tasks relevant to the 4-qubit circuits. Each panel corresponds to a different observable (indicated within the panel), and shows the DAF values for each estimation method for $n_\mathrm{trot}\in\{1,\dots,15\}$. The horizontal dotted lines correspond to sampling overhead reduction percentages of $0\%$, $50\%$, $75\%$, and $90\%$.}
 \label{fig:4-qubits-amp}
\end{figure}
\begin{figure}
 \centering
 \includegraphics[width=\textwidth]{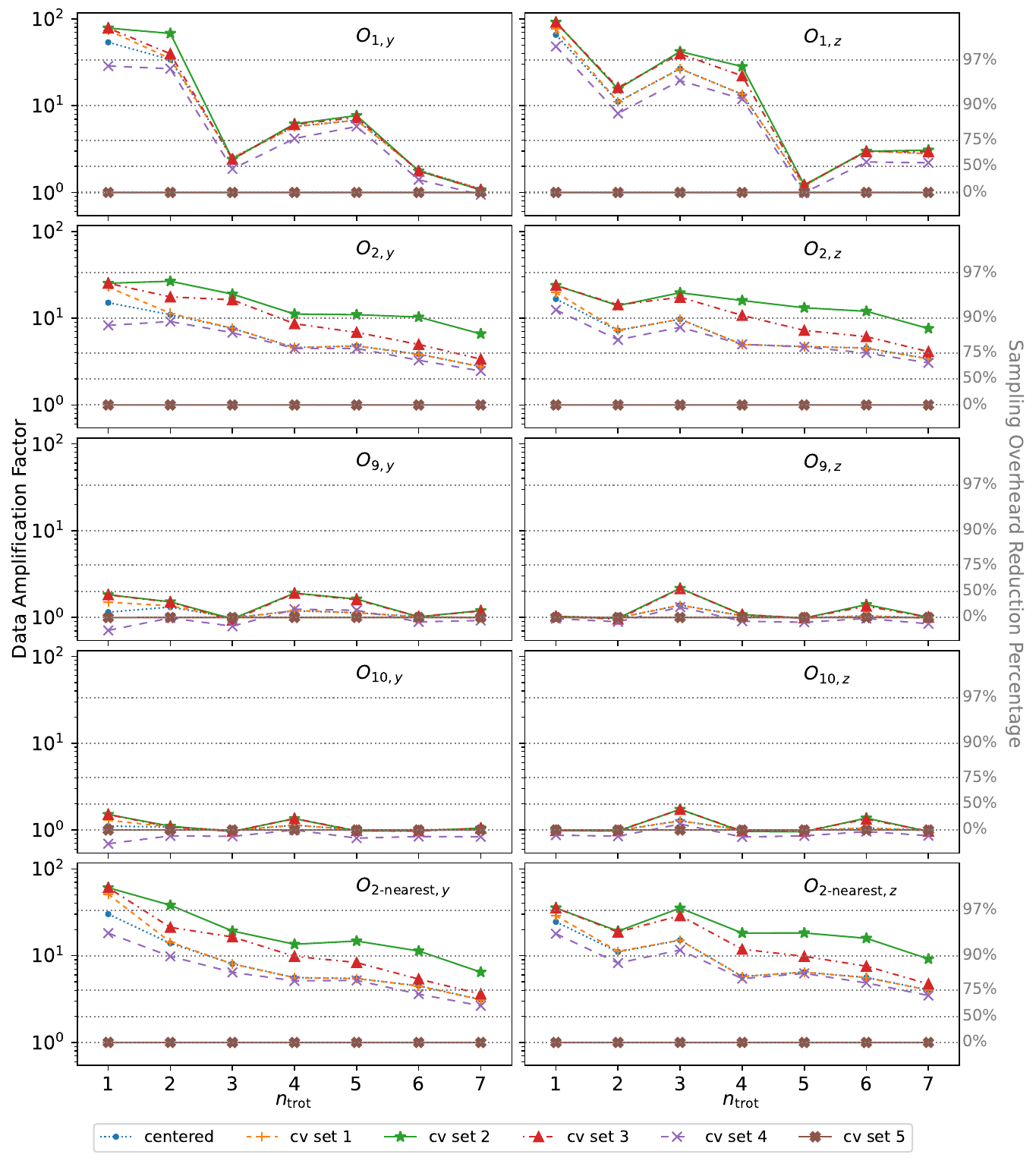}
 \caption{Similar to \fref{fig:4-qubits-amp}, but for the estimation tasks relevant to the 10-qubit circuits. For brevity, panels for only 10 of the 22 observables are known. The horizontal dotted lines correspond to sampling overhead reduction percentages of $0\%$, $50\%$, $75\%$, $90\%$, and $97\%$.}
 \label{fig:10-qubits-amp}
\end{figure}
Some specific observations from \fref{fig:4-qubits-amp} and \fref{fig:10-qubits-amp} are as follows:
\begin{itemize}
 \item The variance reduction achieved by a given CV set depends on the estimation task at hand.
 \item As anticipated in \sref{sec:experiments_cvs}, cv set 5 achieves little to no variance reduction.
 \item For several estimation tasks, just using $W$ as a control (CV set 1) offers more than a 50\% reduction in the sampling overhead.
 \item In most estimation tasks, the centered estimator achieves a variance reduction comparable to CV set 1.
 \item CV sets 2 and 3 outperform CV set 1 for several estimation tasks.
 \item CV set 4 does not significantly outperform CV set 1 in a number of estimation tasks, despite the latter being a subset of the former. However, in the estimation tasks corresponding to $O_{2,y}$, $O_{2,z}$, $O_{\text{2-nearest},y}$, and $O_{\text{2-nearest},z}$, CV set 4 does provide a higher sampling overhead reduction than CV set 1.
\end{itemize} 
The overall variance reduction performance of the different methods is summarized in \tref{tab:performance_summary}. In particular, CV sets 2 and 3 offer a sampling overhead reduction of at least 63\% and 58\%, respectively, in at least 50\% of the estimation tasks.\footnote{A conservative version of this statistic is quoted in the abstract of this paper.} In at least 10\% of the estimation tasks, they offered a sampling overhead reduction of at least 91\% and 89\%, respectively.
\begin{table}[ht]
 \centering
 \caption{The $25^\mathrm{th}$, $50^\mathrm{th}$, $75^\mathrm{th}$, and $90^\mathrm{th}$ percentiles of the data amplification factor values (and sampling overhead reduction percentage values in parentheses) achieved by each of the estimation methods over the 304 estimation tasks. To illustrate how to interpret this table, the highlighted cell means that in at least $(100-75=)$ 25\% of the 304 estimations, cv set 2 gave an amplification factor of at least 5.27 and a corresponding sampling overhead reduction of at least 81\% (the ``at least''-s are indicative of the direction of rounding).}
 \label{tab:performance_summary}
 \small
 \begin{tabular}{M{.15\textwidth} M{.15\textwidth} M{.15\textwidth} M{.15\textwidth} M{.15\textwidth}}
 \toprule
 \multirow{2}*{Method} & \multicolumn{4}{c}{Data Amplification Factor (Sampling Overhead Reduction Percentage)}\\
 \cmidrule(lr){2-5}
 & $25^\mathrm{th}$ percentile  & $50^\mathrm{th}$ percentile  & $75^\mathrm{th}$ percentile  & $90^\mathrm{th}$ percentile \\
 \midrule[.8pt]
 centered & 1.13 (11\%) & 1.67 (40\%) & 2.87 (65\%) & 6.24 (83\%) \\
 \cmidrule(lr){1-5}
 cv set 1 & 1.17 (14\%) & 1.79 (44\%) & 3.32 (69\%) & 6.60 (84\%) \\
 \cmidrule(lr){1-5}
 cv set 2 & 1.39 (28\%) & 2.76 (63\%) & \cellcolor{red!20} 5.27 (81\%) & 11.89 (91\%) \\
 \cmidrule(lr){1-5}
 cv set 3 & 1.35 (25\%) & 2.43 (58\%) & 4.51 (77\%) & 9.71 (89\%) \\
 \cmidrule(lr){1-5}
 cv set 4 & 1.05 (4\%) & 1.80 (44\%) & 3.06 (67\%) & 5.75 (82\%) \\
 \cmidrule(lr){1-5}
 cv set 5 & 0.99 (-1\%) & 1.00 (0\%) & 1.00 (0\%) & 1.00 (0\%) \\
 % \cmidrule(lr){1-5}
 % golden ticket cv & 4.55 (78\%) & 35.03 (97\%) & 145.43 (99.3\%) & 405.54 (99.7\%) \\
 \bottomrule
 \end{tabular}
\end{table}

\section{Summary, Conclusions, and Outlook} \label{sec:conclusions}

In this work, we have proposed using the control variates variance reduction technique to decrease the sampling overhead in quantum computations based on quasi-probabilistic decompositions. We adapted the standard control variates technique in a few ways, in order to make the technique more appropriate for QPD-based estimations. We also described a generic way to construct control variates for use with our technique. Collectively, the techniques introduced in this paper have been dubbed ``CV4Quantum''. We applied our method to probabilistic error cancellation in a realistic simulations-based study, and found significant reductions in sampling overheads in many of the PEC estimations considered. In the rest of this section, we will discuss some directions for future research.

As discussed in this paper, the performance gain achieved by our technique depends on the effectiveness of the control variates used for the particular estimation task at hand. The construction of effective controls to use with our technique could be an impactful line of research. To motivate this, one can consider the ceiling on the variance reduction achievable using the control variates technique. As discussed in \sref{sec:golden-control}, the a priori unknown ``golden-ticket'' control variate can completely eliminate the inter-operation variance term in \eqref{eq:variance_decomposition}. The corresponding performance improvement can be translated into a ``ceiling DAF'' and a ``ceiling SORP'' value.\footnote{This was done by estimating the intra- and inter- operation variances from the experimental data.} In our study, the $25^\mathrm{th}$, $50^\mathrm{th}$, $75^\mathrm{th}$, and $90^\mathrm{th}$ percentiles (over the 304 estimation tasks) of the ceiling DAF (ceiling SORP) values are 4.55 (78\%), 35.03 (97\%), 145.43 (99.3\%), 405.54 (99.7\%), respectively. Comparing this to \tref{tab:performance_summary}, we see that there is room to improve the variance reduction by refining the choice of control variates. This is understandable because a) the inter-operation variance is the dominant source of uncertainty in our study and b) refining the choice of controls for each estimation task was not a goal of this paper.

In this paper, we restricted ourselves to controls whose expectations may be exactly precomputed. However, the control variates technique may also be used with controls whose expectations are only approximately known (e.g., estimated using Monte Carlo sampling), as long as the uncertainties on the means of the controls are comparatively small, as explained in Section 3.2 of Ref.~\cite{doi:https://doi.org/10.1002/9781118445112.stat07975}. This opens up the possibility of using techniques like machine learning to construct/train functions of the indices $(k_1,\dots,k_M)$ that are highly correlated with the observable whose expectation is to be estimated.

A seemingly simple strategy for finding effective controls is the following: Perform the estimation using multiple sets of control variates, and use the estimate  $\Test{cv}$ with the smallest associated uncertainty $\sigmaest{cv}$. However, since $\sigmaest{cv}$ is only an \emph{estimate} (with random fluctuations) for the standard deviation of $\Test{cv}$, picking the minimum $\sigmaest{cv}$ over several sets of control variates could lead to underestimated uncertainties due to the look-elsewhere effect. One should be wary of this when using any approach to construct or distill good controls using data, including the ML-based approach and the try-multiple-sets-of-controls approach described here. A foolproof procedure would be to first use a subset of the available experimental data to construct good controls, and then use these controls with the rest of the experimental data to perform the actual estimation, without a look-elsewhere effect.

The CV4Quantum as described in this paper works only for reducing the inter-operation variance in \eqref{eq:variance_decomposition}. However, it can be extended to tackle the intra-operation variance as well, by creating controls that are not deterministic functions of the indices $(k_1,\dots,k_M)$. More specifically, in addition to the observable of interest $X$, one can have auxiliary observables, say $Y_1, Y_2,\dots$, whose noiseless expectations are exactly known. These auxiliary observables can be measurements on some auxiliary qubits whose final state, at the end of the noiseless computation, is exactly known. They can also be some special combinations (with known noiseless expectations) of the measurements of the primary qubits. Using such $Y_i$-s as controls, one can potentially reduce the intra-operation variance. Exploring this idea is left for future work.

\begin{acks}
 The authors thank Ewout van den Berg, Zlatko K. Minev, Abhinav Kandala, Kristan Temme for providing detailed information regarding the noise models used in their work, Ref.~\cite{vandenBerg2023}.
 %, and for providing feedback on this manuscript.
 The authors thank Ewout van den Berg, Abhinav Kandala, Andrew Eddins, Gabriel Perdue, Stephen Mrenna,
 %Andy C. Y. Li,
 Henry Lamm, Konstantin T. Matchev, and Kyoungchul Kong for useful discussions and/or providing feedback on this manuscript. The following open source software were used \emph{directly} in performing this research and generating images and plots: Python \cite{van1995python}, Qiskit \cite{Qiskit}, NumPy \cite{harris2020array}, SciPy \cite{2020SciPy-NMeth}, Matplotlib \cite{Hunter:2007}, Jupyterlab \cite{Jupyterlab}, quantikz \cite{kay2023tutorial}.
 
 % This manuscript has been authored by Fermi Research Alliance, LLC under Contract No. DE-AC02-07CH11359 with the U.S. Department of Energy, Office of Science, Office of High Energy Physics.

 This document was prepared using the resources of the Fermi National Accelerator Laboratory (Fermilab), a U.S. Department of Energy, Office of Science, Office of High Energy Physics HEP User Facility. Fermilab is managed by Fermi Forward Discovery Group, LLC, acting under Contract No. 89243024CSC000002.
 
 PS is supported by the U.S. Department of Energy, Office of Science, Office of High Energy Physics QuantISED program under the grants ``HEP Machine Learning and Optimization Go Quantum'', Award Number 0000240323, and ``DOE QuantiSED Consortium QCCFP-QMLQCF'', Award Number DE-SC0019219. WY was supported by the National Science Foundation Mathematical Sciences Graduate Internship (MSGI) Program in summer 2022.
 
 PS thanks the organizers of the workshop titled ``Quantum Error Mitigation for Particle and Nuclear Physics'' at the InQubator for Quantum Simulation (IQuS), University of Washington, Seattle, where this work was conceived. This work was partially performed at the Aspen Center for Physics (ACP), which is supported by National Science Foundation grant PHY-1607611. PS thanks ACP for hospitality and support during May--June of 2022 and 2023.

 Any opinions, findings, and conclusions or recommendations expressed in this material are those of the authors and do not necessarily reflect the views of the National Science Foundation or the U.S. Department of Energy.
\end{acks}

\section*{Code and Data Availability}
The code and data that support the findings of this study are openly available at the following
URL: \url{https://gitlab.com/prasanthcakewalk/cv4quantum-demo}.

\bibliographystyle{unsrtnat}
\bibliography{references}

\appendix
\section{Lemmas and Proofs} \label{appendix:lemmas_and_proofs}
The different symbols used in this section have the same meanings as in the rest of this paper. From the properties of the Moore--Penrose inverse, we have
\begin{align}
 \sum_{b,c=1}^{N_\mathrm{cv}}K^+_{ab}\,K_{bc}\,K^+_{cd} &= K^+_{ad}\,, \label{eq:2Kplus1K}\\
 \sum_{b,c=1}^{N_\mathrm{cv}}K_{ab}\,K^+_{bc}\,K_{cd} &= K_{ad}\,. \label{eq:1Kplus2K}
\end{align}
Furthermore, $\mathbf{K}$ and $\mathbf{K}^+$ are both symmetric matrices.

\begin{lemma} \label{lemma:span}
If there exist constants $\alpha_1,\dots, \alpha_{N_\mathrm{cv}}$ such that
\begin{align}
 U - \E\left[U\right] = \sum_{a=1}^{N_\mathrm{cv}} \alpha_a\,\left(\big.V_a - \mu_a\right)\,, \label{eq:U_in_span}
\end{align}
then $\var\left[U_{(\lambda^\ast_1,\dots, \lambda^\ast_{N_\mathrm{cv}})}\right] = 0$.
\end{lemma}
\begin{proof}
From \eqref{eq:var_U_lambdas}, \eqref{eq:lambda_star}, and \eqref{eq:2Kplus1K} we have
\begin{align}
% \begin{split}
%  \var\left[U_{(\lambda^\ast_1,\dots, \lambda^\ast_{N_\mathrm{cv}})}\right] &= \var\left[U\right] - 2\,\sum_{a,b=1}^{N_\mathrm{cv}}K^+_{ab}\,\cov\left[U, V_a\right]\,\cov\left[U, V_b\right] \\
%  &\qquad\qquad+ \sum_{a,b,c,d=1}^{N_\mathrm{cv}}K^+_{ab}\,K^+_{cd}\,K_{ac}\,\cov\left[U, V_b\right]\,\cov\left[U, V_d\right]
% \end{split}\\
 \var\left[U_{(\lambda^\ast_1,\dots, \lambda^\ast_{N_\mathrm{cv}})}\right] &= \var\left[U\right] - \sum_{a,b=1}^{N_\mathrm{cv}}K^+_{ab}\,\cov\left[U, V_a\right]\,\cov\left[U, V_b\right]
\end{align}
Substituting the expression for $U$ from \eqref{eq:U_in_span} here, and using Bienaym\'{e}'s identity and \eqref{eq:1Kplus2K}, we have 
\begin{align}
 \var\left[U_{(\lambda^\ast_1,\dots, \lambda^\ast_{N_\mathrm{cv}})}\right] &= \sum_{a,b=1}^{N_\mathrm{cv}}K_{ab}\,\alpha_a\,\alpha_b - \sum_{a,b,c,d=1}^{N_\mathrm{cv}} K^+_{ab}\,K_{ac}\,\alpha_c\,K_{bd}\,\alpha_d \qquad = 0
\end{align}
\end{proof}

\begin{lemma} \label{lemma:goldenticket}
If $~U\equiv WX$, and $V\equiv W\,\E\left[X\,\big|\, k_1,\dots,k_M\right]$, then
\begin{align}
 \var\left[U_{\lambda^\ast}\right] = \E\left[\Big.W^2\,\var\left[X\,\big|\, k_1,\dots,k_M\right]\right]\,.
\end{align}
\end{lemma}
\begin{proof}
Using the law of total covariance, we have
\begin{align}
 \cov[U, V] &= \E\left[\Big.\cov\left[U, V\,\big|\,k_1,\dots,k_M\right]\right] + \cov\left[\Big.\E\left[U\,\big|\,k_1,\dots,k_M\right]~~,~~\E\left[V\,\big|\,k_1,\dots,k_M\right]\right]\,.\label{eq:lemma_proof_1}
\end{align}
Since $V$ is a deterministic function of $(k_1,\dots,k_M)$, we have
\begin{align}
 \var\left[V\,\big|\,k_1,\dots,k_M\right] = \cov\left[U, V\,\big|\,k_1,\dots,k_M\right] = 0\,.\label{eq:lemma_proof_2}
\end{align}
Furthermore, since $W$ is a deterministic function of $(k_1,\dots,k_M)$
\begin{align}
 \E\left[U\,\big|\,k_1,\dots,k_M\right] = \E\left[V\,\big|\,k_1,\dots,k_M\right]\,.\label{eq:lemma_proof_3}
\end{align}
From \eqref{eq:lemma_proof_1}, \eqref{eq:lemma_proof_2}, and \eqref{eq:lemma_proof_3}, and the law of total variance we have
\begin{alignat}{2}
 \cov[U, V] ~~&=~~ \var\left[\Big.\E\left[U\,\big|\,k_1,\dots,k_M\right]\right] ~~&&=~~ \var\left[\Big.\E\left[V\,\big|\,k_1,\dots,k_M\right]\right]\\
 ~~&=~~ \var[U] - \E\left[\Big.\var\left[U\,\big|\,k_1,\dots,k_M\right]\right] ~~&&=~~ \var[V]\,. \label{eq:lemma_proof_4}
\end{alignat}
Now, from \eqref{eq:min_var_1cv}, \eqref{eq:lemma_proof_4}, and the fact that $W$ is a deterministic function of $(k_1,\dots,k_M)$, we have
\begin{align}
 \var[U_{\lambda^\ast}] &= \var[U] - \var[V]\\
 &= \E\left[\Big.\var\left[U\,\big|\,k_1,\dots,k_M\right]\right]\\
 &= \E\left[\Big.W^2\,\var\left[X\,\big|\, k_1,\dots,k_M\right]\right]\,.
\end{align}
\end{proof}

\begin{lemma} \label{lemma:shiftinvariance}
If the following two assumptions are both true, then {\normalfont$\Test{unbiased-cv}$} defined in \eqref{eq:Test_unbiased_cv} is shift invariant.
\begin{enumerate}
 \item[] Assumption 1:\quad$(W-\mu_W)$ lies in the span of the random variables $(V_a-\mu_a)$. 
 \item[] Assumption 2:\quad$\lambda^{(\neg i)}_a$ transforms as $\displaystyle\lambda^{(\neg i)}_a \mapsto \lambda^{(\neg i)}_a + \delta \sum_{b=1}^{N_\mathrm{cv}}K^+_{ab}\,C_b$ under $X^{(i)}\mapsto X^{(i)}+\delta$.
\end{enumerate}
\end{lemma}
\begin{proof}
As in \alref{alg:cv_estimator}, let us define $W_\text{res}^{(i)}$ as
\begin{align}
 W_\text{res}^{(i)} \equiv W^{(i)} - \sum_{a,b=1}^{N_\mathrm{cv}} K^+_{ab}\,C_b\,\left(V_a^{(i)}-\mu_a\right)\,.
\end{align}
From Assumption 1 and \lref{lemma:span}, we have $\var\left[W_\text{res}^{(i)}\right] = 0$. Since $W_\text{res}^{(i)}$ is drawn from a discrete probability distribution, this implies that $W_\text{res}^{(i)} = \E\left[W_\text{res}^{(i)}\right] = \mu_W$\,. Now, from Assumption 2 and \eqref{eq:Test_unbiased_cv}, one can see that under the transformation $X^{(i)}\mapsto X^{(i)}+\delta$, $\Test{unbiased-cv}$ transforms as
\begin{align}
 \Test{unbiased-cv} \mapsto \Test{unbiased-cv} + \frac{\delta}{N}\sum_{i=1}^N W_\text{res}^{(i)} = \Test{unbiased-cv} + \delta\,\mu_W\,,
\end{align}
\end{proof}

\section{Error Estimation Formulas} \label{appendix:uncertainty_estimates}

This section contains the derivations for the error estimation formulas provided in \sref{sec:formulas}. Consider a generic estimator $\Test{method}$ for $T$ given by
\begin{align}
 \var\left[\Test{method}\right] &\equiv \frac{1}{N}\sum_{i=1}^N Y^{(\sim i)}_\mathrm{method}
\end{align}
where $Y^{(\sim i)}_\mathrm{method}$-s are exchangeable random variables. Using Bienaym\'{e}'s identity and the exchangeability of $Y^{(\sim i)}_\mathrm{method}$-s, the variance of $\Test{method}$ is given by
\begin{align}
 \var\left[\Test{method}\right] &= \frac{1}{N}\,\var\left[Y^{(\sim 1)}_\mathrm{method}\right] + \frac{N-1}{N}\,\cov\left[Y^{(\sim 1)}_\mathrm{method}~,~Y^{(\sim 2)}_\mathrm{method}\right]\,.\label{eq:tmethod_generic}
\end{align}
Next, let us consider the expected value of the sample variance of $Y^{(\sim i)}_\mathrm{method}$-s:
\begin{align}
 \E\left[\Big.\svar\left[Y_\mathrm{method}\right]\right] &= \frac{1}{2}~\E\left[\left(Y^{(\sim 1)}_\mathrm{method} - Y^{(\sim 2)}_\mathrm{method}\right)^2\right]\\
 &= \E\left[\left(Y^{(\sim 1)}_\mathrm{method}\right)^2\right] - \E\left[Y^{(\sim 1)}_\mathrm{method}~Y^{(\sim 2)}_\mathrm{method}\right]\\
 &= \var\left[Y^{(\sim 1)}_\mathrm{method}\right] - \cov\left[Y^{(\sim 1)}_\mathrm{method}~,~Y^{(\sim 2)}_\mathrm{method}\right]\,.\label{eq:var_1stterm_expectation}
\end{align}
From \eqref{eq:tmethod_generic} and \eqref{eq:var_1stterm_expectation}, we have
\begin{align}
 \var\left[\Test{method}\right] - \E\left[\frac{1}{N}\svar\left[Y_\mathrm{method}\right]\right] = \cov\left[Y^{(\sim 1)}_\mathrm{method}~,~Y^{(\sim 2)}_\mathrm{method}\right]\,.\label{eq:var_Tmethod_generic_secondterm}
\end{align}
This means that we can estimate $\var\left[\Test{method}\right]$ as $\svar\left[Y_\mathrm{method}\right]/N$ plus a correction term to estimate or upperbound $\cov\left[Y^{(\sim 1)}_\mathrm{method}~,~Y^{(\sim 2)}_\mathrm{method}\right]$. Next we will derive this correction term for the centered and cv estimators.

\subsection{Error Estimation Formula for the Centered Estimator}
From \eqref{eq:alg2_b}, using the distributive property of covariance over addition, the fact that $(W^{(i)}, X^{(i)})$-s are independent and identically distributed (iid), and the fact that $\E[W^{(i)}] = \mu_W$ we have
\begin{align}
 \cov\left[Y^{(\sim 1)}_\mathrm{centered}~,~Y^{(\sim 2)}_\mathrm{centered}\right] &= \frac{1}{(N-1)^2}\,\cov\left[X^{(2)}\,\left(W^{(1)}-\mu_W\right)~,~X^{(1)}\,\left(W^{(2)}-\mu_W\right)\right]\\
 &= \frac{1}{(N-1)^2}\E^2\left[X^{(1)}\,\left(W^{(1)}-\mu_W\right)\right]\\
 &= \frac{1}{(N-1)^2}\cov^2\left[X^{(1)}, W^{(1)}\right]\,.\label{eq:cov_Y_centered}
\end{align}
$\cov^2\left[X^{(1)},W^{(1)}\right]$ can be estimated as $\scov^2\left[X,W\right]$, leading to the formula for $\sigmaest{centered}$ in \eqref{eq:alg2_e}. Furthermore, since $\scov\left[X,W\right]$ is an unbiased estimator for $\cov\left[X^{(1)},W^{(1)}\right]$, we have
\begin{align}
 \E\left[\Big.\scov^2\left[X,W\right]\right] &= \E^2\left[\Big.\scov\left[X,W\right]\right] + \var\left[\Big.\scov\left[X, W\right]\right]\\
 &= \cov^2\left[X^{(1)}, W^{(1)}\right] + \left|~\mathcal{O}\left(\frac{1}{N}\right)~\right|\,,\label{eq:scov_X_W_bias}
\end{align}
where $\mathcal{O}$ corresponds to the big-O notation. From \eqref{eq:alg2_e}, \eqref{eq:var_Tmethod_generic_secondterm}, \eqref{eq:cov_Y_centered}, and \eqref{eq:scov_X_W_bias}, we have
\begin{align}
 \E\left[\sigmaest{centered}^2\right] &= \var\left[\Test{centered}\right] + \left|~\mathcal{O}\left(\frac{1}{N^3}\right)~\right|\,.\label{eq:bias_centered}
\end{align}
This is one of the bias properties in \tref{tab:estimator_properties}.

\subsection{Error Estimation Formula for the Control-Variates-Based Estimator}

From \eqref{eq:alg3_c}, \eqref{eq:alg3_d}, and \eqref{eq:alg3_f}, $Y^{(\sim i)}_\mathrm{cv}$ can be written as
\begin{align}
 Y^{(\sim i)}_\mathrm{cv} &= W^{(i)}\,X^{(i)} - \sum_{a,b=1}^{N_\mathrm{cv}}K^+_{ab}\,\left(S^{(\neg i)}_a + C_a\,\smean^{(\neg i)}_\text{sans-one}[X]\right)\,\left(V_b^{(i)} - \mu_b\right)\\
 &= W^{(i)}\,X^{(i)} - \frac{1}{2(N-1)(N-2)}\sum_{a,b=1}^{N_\mathrm{cv}}\sum_{j\neq i}\sum_{k\neq i\,\text{or}\,j}K^+_{ab}\,\beta^{(jk)}_a\,\left(V_b^{(i)} - \mu_b\right)
\end{align}
where $\beta_a^{(jk)}$ is symmetric in the indices $(j,k)$ and is given by
\begin{align}
 \beta_a^{(jk)} &\equiv \left(X^{(j)} - X^{(k)}\right)\left(G^{(j)}_a - G^{(k)}_a\right) + \left(X^{(j)} + X^{(k)}\right)\,C_a\,.
\end{align}
Now, using the distributive property of covariance over addition, the fact that $(W^{(i)}, X^{(i)})$-s are iid, and the fact that $\E\left[V^{(i)}_b\right] = \mu_b$ we have
\begin{align}
\begin{split}
 &\cov\left[Y^{(\sim 1)}_\mathrm{cv}~,~Y^{(\sim 2)}_\mathrm{cv}\right] = \frac{1}{(N-1)^2(N-2)^2}\cov\left[\sum_{a,b=1}^{N_\mathrm{cv}}\sum_{i=3}^N K^+_{ab}\,\beta^{(2i)}_a\,\left(V_b^{(1)} - \mu_b\right)~,~\right.\\
 &\qquad\qquad\qquad\qquad\qquad\qquad\qquad\qquad\qquad\qquad\left.\sum_{c,d=1}^{N_\mathrm{cv}}\sum_{j=1}^N K^+_{cd}\,\beta^{(1j)}_c\,\left(V_d^{(2)} - \mu_d\right)\right]
\end{split}\\
\begin{split}
 &= \frac{1}{(N-1)^2(N-2)^2}\cov\left[\sum_{a,b=1}^{N_\mathrm{cv}}\sum_{i=3}^N K^+_{ab}\,\left(\beta^{(2i)}_a-\E\left[\beta^{(12)}_a\right]\right)\,\left(V_b^{(1)} - \mu_b\right)~,~\right.\\
 &\qquad\qquad\qquad\qquad\qquad\qquad\left.\sum_{c,d=1}^{N_\mathrm{cv}}\sum_{j=3}^N K^+_{cd}\,\left(\beta^{(1j)}_c - \E\left[\beta^{(12)}_c\right]\right)\,\left(V_d^{(2)} - \mu_d\right)\right]
\end{split}
\end{align}
Using the fact that $\cov[A, B] \leq \var[A] = \var[B]$ for two identically distributed (not necessarily independent) random variables $A$ and $B$, we have
\begin{align}
\begin{split}
 \cov&\left[Y^{(\sim 1)}_\mathrm{cv},Y^{(\sim 2)}_\mathrm{cv}\right]\leq \frac{1}{(N-1)^2(N-2)^2}\cov\left[\sum_{a,b=1}^{N_\mathrm{cv}}\sum_{i=3}^N K^+_{ab}\,\left(\beta^{(1i)}_a-\E\left[\beta^{(12)}_a\right]\right)\,\left(V_b^{(2)} - \mu_b\right)~,~\right.\\
 &\qquad\qquad\qquad\qquad\qquad\qquad\qquad\qquad\left.\sum_{c,d=1}^{N_\mathrm{cv}}\sum_{j=3}^N K^+_{cd}\,\left(\beta^{(1j)}_c - \E\left[\beta^{(12)}_c\right]\right)\,\left(V_d^{(2)} - \mu_d\right)\right]
\end{split}\label{eq:Ycv_cov_var_ineq}\\
 &= \frac{1}{(N-1)^2(N-2)^2}\sum_{a,b,c,d=1}^{N_\mathrm{cv}}K^+_{ab}\,K_{bd}\,K^+_{cd}\,\cov\left[\sum_{i=3}^N \beta_a^{(1i)}~,~\sum_{j=3}^N\beta_c^{(1j)}\right]\\
 &= \frac{1}{(N-1)^2(N-2)^2}\sum_{a,b=1}^{N_\mathrm{cv}}K^+_{ab}\,\cov\left[\sum_{i=3}^N \beta_a^{(1i)}~,~\sum_{j=3}^N\beta_b^{(1j)}\right] \label{eq:cov_intermsof_betasum}\\
\begin{split}
 &= \frac{1}{(N-1)^2}\sum_{a,b=1}^N K^+_{ab}\left(\bigg.\frac{N-3}{N-2}\,\cov\left[\beta^{(12)}_a~,~\beta^{(13)}_b\right]+~\frac{1}{N-2}\,\cov\left[\beta^{(12)}_a~,~\beta^{(12)}_b\right]\right) \label{eq:cov_intermsof_beta}
\end{split}
\end{align}
Equation \eqref{eq:cov_intermsof_betasum} suggests upperbounding $\cov\left[Y^{(\sim 1)}_\mathrm{cv}, Y^{(\sim 2)}_\mathrm{cv}\right]$ via the quantity $F^{(\sim i)}_a$ defined as follows
\begin{align}
 F^{(\sim i)}_a &\equiv \frac{1}{N-1}\sum_{j\neq i}\beta_a^{(ij)}\,.
\end{align}
It can be shown that $F^{(\sim i)}_a$ is identical to the $L_a^{(\sim i)}$ defined in \alref{alg:cv_estimator} in equations \eqref{eq:alg3_a} and \eqref{eq:alg3_b}, up to an additive term independent of $i$ as follows:
\begin{align}
 F^{(\sim i)}_a &= \frac{1}{N-1}\sum_{j\neq i} \left[\left(R^{(\sim i)} - R^{(\sim j)}\right)\left(G_a^{(i)} - G_a^{(j)}\right) + \left(R^{(\sim i)} + R^{(\sim j)}\right)\,C_a + 2\,\smean[X]\,C_a\right]\\
\begin{split}
 &= \frac{1}{N-1}\sum_{j\neq i} \bigg[R^{(\sim i)}\,\left(G_a^{(i)}-\smean\left[G_a\right]\right) + R^{(\sim j)}\,\left(G_a^{(j)}-\smean\left[G_a\right]\right)\\
 &\qquad\qquad\qquad\qquad - R^{(\sim i)}\,\left(G_a^{(j)}-\smean\left[G_a\right]\right) - R^{(\sim j)}\,\left(G_a^{(i)}-\smean\left[G_a\right]\right)\\
 &\qquad\qquad\qquad\qquad + \left(R^{(\sim i)} + R^{(\sim j)}\right)\,C_a + 2\,\smean[X]\,C_a\bigg]
\end{split}\\
\begin{split}
 &= \frac{1}{N-1}\bigg[(N-1)\,R^{(\sim i)}\,\left(G_a^{(i)}-\smean\left[G_a\right]\right) - R^{(\sim i)}\,\left(G_a^{(i)}-\smean\left[G_a\right]\right)\\
 &\qquad\qquad\qquad\qquad+ R^{(\sim i)}\,\left(G_a^{(i)}-\smean\left[G_a\right]\right) + R^{(\sim i)}\,\left(G_a^{(i)}-\smean\left[G_a\right]\right)\\
 &\qquad\qquad\qquad\qquad + \left((N-1)\,R^{(\sim i)} - R^{(\sim i)}\right)\,C_a\bigg] + (\text{term independent of }i)
\end{split}\\
 &= \frac{R^{(\sim i)}}{N-1}\left[(N-2)\,C_a + N\,\left(G_a^{(i)}-\smean\left[G_a\right]\right)\right] + (\text{term independent of }i)\\
 &= L_a^{(\sim i)} + (\text{term independent of }i)\,. \label{eq:F_L_connection}
\end{align}
Let us consider the expected value of the sample covariance of $L^{(\sim i)}_a$ and $L^{(\sim i)}_b$. Using \eqref{eq:F_L_connection}, the distributive property of covariance over addition and the fact that $\beta_a^{(ij)}$ is symmetric in $(i,j)$, we have
\begin{align}
 \E&\left[\Big.\scov\left[L_a, L_b\right]\right] =\E\left[\Big.\scov\left[F_a, F_b\right]\right] = \frac{1}{2}~\E\left[\left(F_a^{(\sim 1)} - F_a^{(\sim 2)}\right)\left(F_b^{(\sim 1)} - F_b^{(\sim 2)}\right)\right] \\
 &= \E\left[F_a^{(\sim 1)}\,F_b^{(\sim 1)}\right] - \E\left[F_a^{(\sim 1)}\,F_b^{(\sim 2)}\right] \\
 &= \cov\left[F^{(\sim 1)}_a,F^{(\sim 1)}_b\right] - \cov\left[F^{(\sim 1)}_a,F^{(\sim 2)}_b\right]\\
\begin{split}
 &= \frac{N-2}{N-1}\,\cov\left[\beta_a^{(12)}~,~\beta_b^{(13)}\right] + \frac{1}{N-1}\,\cov\left[\beta_a^{(12)}~,~\beta_b^{(12)}\right]\\
 &\qquad- \frac{3(N-2)}{(N-1)^2}\,\cov\left[\beta_a^{(12)}~,~\beta_b^{(13)}\right] - \frac{1}{(N-1)^2}\,\cov\left[\beta_a^{(12)}~,~\beta_b^{(12)}\right]
\end{split}\\
 &= \frac{(N-2)(N-3)}{(N-1)^2}\left(\frac{N-4}{N-3}\,\cov\left[\beta_a^{(12)}~,~\beta_b^{(13)}\right] + \frac{1}{N-3}\,\cov\left[\beta_a^{(12)}~,~\beta_b^{(12)}\right]\right)\,.\label{eq:scov_L}
\end{align}
From \eqref{eq:cov_intermsof_beta} and \eqref{eq:scov_L} we have
\begin{align}
\begin{split}
 \cov&\left[Y^{(\sim 1)}_\mathrm{cv}~,~Y^{(\sim 2)}_\mathrm{cv}\right] - \E\left[\frac{1}{(N-2)(N-3)}\,\sum_{a,b=1}^{N_\mathrm{cv}}K^+_{ab}\,\scov[L_a, L_b]\right] \\
 &\leq \frac{1}{(N-1)^2(N-2)(N-3)}\sum_{a,b=1}^{N_\mathrm{cv}}K^+_{ab}\left(\cov\left[\beta^{(12)}_a~,~\beta_b^{(13)}\right] - \cov\left[\beta^{(12)}_a~,~\beta_b^{(12)}\right]\right)
\end{split}\\
 &= -~\frac{1}{2(N-1)^2(N-2)(N-3)}\sum_{a,b=1}^{N_\mathrm{cv}}K^+_{ab}\,\underbrace{\cov\left[\beta^{(12)}_a-\beta^{(13)}_a~,~\beta_b^{(12)}-\beta_b^{(13)}\right]}_{\text{positive semidefinite matrix indexed by } (a,b)}\\
 &\leq 0\qquad\qquad\left(\genfrac{}{}{0pt}{0}{\text{\small since the Frobenius inner product of two}}{\text{\small positive semidefinite matrixes is non-negative}}\right)\,.\label{eq:cov_Ycv_ineq}
\end{align}
From \eqref{eq:alg3_i}, \eqref{eq:var_Tmethod_generic_secondterm}, \eqref{eq:cov_Ycv_ineq}, and the nature of the inequality in \eqref{eq:Ycv_cov_var_ineq}, we have
\begin{align}
 \E\left[\sigmaest{cv}^2\right] &= \var\left[\Test{cv}\right] + \left|~\mathcal{O}\left(\frac{1}{N^2}\right)~\right|\,.\label{eq:bias_cv}
\end{align}
This is one of the bias properties in \tref{tab:estimator_properties}.

\section{Proofs of Some Properties of the Estimators} \label{appendix:estimator_properties}

Here we will prove the properties in \tref{tab:estimator_properties} that have not been proved in \sref{sec:methodology} and/or \aref{appendix:uncertainty_estimates}.

\begin{lemma}
 $\sigmaest{cv}^2$ in \eqref{eq:alg3_i} is non-negative.
\end{lemma}
\begin{proof}
 This lemma follows from the facts that a) sample variances are non-negative, b) sample covariance matrices are positive semidefinite, c) the matrix $\mathbf{K}^+$ is positive semidefinite, and d) the Frobenius inner product of two positive semidefinite matrices is non-negative. 
\end{proof}

\begin{lemma}
 $\Test{cv}$ and $\sigmaest{cv}^2$ are invariant under invertible linear transformations of $V_a$ (with appropriate changes to $\mu_a$, $C_a$, and $K^+_{ab}$).
\end{lemma}
\begin{proof}
Let the vector notation $\vec{A}$ represent an $N_\mathrm{cv}$ dimensional column vector $(A_1,\dots, A_{N_\mathrm{cv}})$. Let $\mathbf{M}$ be an invertible $N_\mathrm{cv}\times N_\mathrm{cv}$ matrix and let $\vec{\alpha}$ be a constant vector. To prove the lemma, we will simply track how different quantities in \alref{alg:cv_estimator} transform under the affine transformation
 \begin{align}
  \vec{V}^{(i)} \longmapsto \mathbf{M}\,\vec{V}^{(i)} + \vec{\alpha}\,.
 \end{align}
 Under this transformation, we have
 \begin{align}
  \vec{\mu} &\longmapsto \mathbf{M}\,\vec{\mu} + \vec{\alpha}\,,\\
  \vec{C} &\longmapsto \mathbf{M}\,\vec{C}\,,\\
  \mathbf{K} &\longmapsto \mathbf{M}\,\mathbf{K}\,\mathbf{M}^T\,,\\
  \mathbf{K}^+ &\longmapsto \left(\mathbf{M}^T\right)^{-1}\,\mathbf{K}^+\,\mathbf{M}^{-1}\,.\label{eq:Kplus_transformation}
 \end{align}
Plugging these into \eqref{eq:alg3_a}, \eqref{eq:alg3_b}, \eqref{eq:alg3_c}, and \eqref{eq:alg3_d}, we have
\begin{align}
 \vec{G}^{(i)} &\longmapsto \mathbf{M}\,\vec{G}^{(i)}\,,\\
 \vec{L}^{(\sim i)} &\longmapsto \mathbf{M}\,\vec{L}^{(\sim i)}\,, \label{eq:L_transformation}\\
 W^{(i)} - W^{(i)}_\mathrm{res} &\longmapsto W^{(i)} - W^{(i)}_\mathrm{res}\,, \label{eq:W_transformation}\\
 \vec{S}_a^{(\neg i)} &\longmapsto \mathbf{M}\,\vec{S}^{(\neg i)}\,. \label{eq:S_transformation}
\end{align}
Plugging \eqref{eq:W_transformation} and \eqref{eq:S_transformation} into \eqref{eq:alg3_f} we have
\begin{align}
 Y_\mathrm{cv}^{(\sim i)} &\longmapsto Y_\mathrm{cv}^{(\sim i)}\,.\label{eq:Y_transformation}
\end{align}
Let $\mathbf{Lscov}$ be an $N_\mathrm{cv}\times N_\mathrm{cv}$ matrix whose $(a,b)$-the element is $\scov\left[L_a, L_b\right]$. From \eqref{eq:Kplus_transformation} and \eqref{eq:L_transformation}, we have
\begin{align}
 \mathbf{Lscov} &\longmapsto \mathbf{M}\,\mathbf{Lscov}\,\mathbf{M}^T\,,\\
 \mathbf{K}^+\,\mathbf{Lscov} &\longmapsto \left(\mathbf{M}^T\right)^{-1}\,\mathbf{K}^+\,\mathbf{Lscov}\,\mathbf{M}^T\,,\\
 \Longrightarrow~~\mathrm{trace}\left[\Big.\mathbf{K}^+\,\mathbf{Lscov}\right] &\longmapsto \mathrm{trace}\left[\Big.\mathbf{K}^+\,\mathbf{Lscov}\right]\,.\label{eq:trace_transformation}
\end{align}
From \eqref{eq:alg3_h}, \eqref{eq:alg3_i}, \eqref{eq:Y_transformation}, and \eqref{eq:trace_transformation} we have
\begin{align}
 \Test{cv} &\longmapsto \Test{cv}\,,\qquad\qquad\qquad\sigmaest{cv}^2 \longmapsto \sigmaest{cv}^2\,.
\end{align}
\end{proof}

\subsection*{Proofs of the scaling properties in \tref{tab:estimator_properties}}

From \eqref{eq:bias_centered} and \eqref{eq:bias_cv} (and the well known properties of the sample variance of the average of iid variables), it follows that
\begin{subequations}\label{eq:bias_sigmasq_method}
\begin{align}
 \text{Bias of }\sigmaest{basic}^2 &= 0\,,\\
 \text{Bias of }\sigmaest{centered}^2 &= \left|\mathcal{O}(N^{-3})\right|\,,\\
 \text{Bias of }\sigmaest{cv}^2 &= \left|\mathcal{O}(N^{-2})\right|\,.
\end{align}
\end{subequations}
In the large $N$ limit, $Y_\mathrm{cv}^{(\sim i)}$ becomes approximately independent of $Y_\mathrm{cv}^{(\sim j)}$ for $i\neq j$ (and likewise for $Y_\mathrm{centered}^{(\sim i)}$). In this limit, the first terms of \eqref{eq:alg2_e} and \eqref{eq:alg3_i} dominate the corresponding second terms. From this, and the fact that the variance of sample variance of iid variables scales as $\mathcal{O}(N^{-1})$, it follows that
\begin{align}
 \E\left[\sigmaest{method}^2\right] &= \mathcal{O}(N^{-1})\,,\qquad\forall \text{method}\in\{\text{basic, centered, cv}\}\,,\label{eq:E_sigmaesq_method}\\
 \var\left[\sigmaest{method}^2\right] &= \mathcal{O}(N^{-3})\,,\qquad\forall \text{method}\in\{\text{basic, centered, cv}\}\,.\label{eq:var_sigmaesq_method}
\end{align}
From \eqref{eq:bias_sigmasq_method} and \eqref{eq:E_sigmaesq_method}, it follows that
\begin{align}
 \var\left[\Test{method}\right] = \mathcal{O}(N^{-1})\,,\qquad\forall \text{method}\in\{\text{basic, centered, cv}\}\,.\label{eq:var_T}
\end{align}
From \eqref{eq:bias_sigmasq_method}, \eqref{eq:var_sigmaesq_method}, and the definition of MSE, it follows that
\begin{align}
 \text{MSE of }\sigmaest{method}^2 = \mathcal{O}(N^{-3})\,,\qquad\forall \text{method}\in\{\text{basic, centered, cv}\}\,.\label{eq:mse_sigmasq_method}
\end{align}
We can write
\begin{align}
 \text{MSE of } \sigmaest{method}^2 &= \E\left[\left(\sigmaest{method}^2 - \var\left[\Test{method}\right]\right)^2\right]\,,\\
 &= \E\left[\left(\sigmaest{method} - \sqrt{\var\left[\Test{method}\right]}\right)^2~\left(\sigmaest{method} + \sqrt{\var\left[\Test{method}\right]}\right)^2\right]\,.\label{eq:mse_relationship}
\end{align}
In the large $N$ limit, $\left(\sigmaest{method} + \sqrt{\var\left[\Test{method}\right]}\right)^2$ scales, like $\var\left[\Test{method}\right]$, as $\mathcal{O}(N^{-1})$. Using this, \eqref{eq:mse_sigmasq_method}, and \eqref{eq:mse_relationship} we have
\begin{align}
 \text{MSE of }\sigmaest{method} = \mathcal{O}(N^{-2})\,,\qquad\forall \text{method}\in\{\text{basic, centered, cv}\}\,.\label{eq:mse_sigma_method}
\end{align}
From the definition of MSE, we have
\begin{align}
 \text{Bias of }\sigmaest{method} = \mathcal{O}(N^{-1})\,,\qquad\forall \text{method}\in\{\text{basic, centered, cv}\}\,,\\
 \var\left[\sigmaest{method}\right] = \mathcal{O}(N^{-2})\,,\qquad\forall \text{method}\in\{\text{basic, centered, cv}\}\,.\label{eq:var_sigma_method}
\end{align}
From Jensen's inequality we have
\begin{align}
 \E\left[\sigmaest{method}^2\right] - \var\left[\Test{method}\right] &\geq \E^2\left[\sigmaest{method}\right] - \var\left[\Test{method}\right]\\
 &= \left(\E\left[\sigmaest{method}\right] - \sqrt{\var\left[\Test{method}\right]}\right)~\left(\E\left[\sigmaest{method}\right] + \sqrt{\var\left[\Test{method}\right]}\right)\,,\\
 \Rightarrow \text{Bias of } \sigmaest{method}^2 &\geq \left(\text{Bias of } \sigmaest{method}\right)~\left(\E\left[\sigmaest{method}\right] + \sqrt{\var\left[\Test{method}\right]}\right)
\end{align}
Since the bias of $\sigmaest{method}^2$ is non-negative and $\sigmaest{method}$ is non-negative, it follows that
\begin{align}
 \text{Bias of }\sigmaest{method} = \left|\mathcal{O}(N^{-1})\right|\,,\qquad\forall \text{method}\in\{\text{basic, centered, cv}\}\,.\label{eq:bias_sigma_method}
\end{align}
Finally, using \eqref{eq:var_T}, \eqref{eq:mse_sigmasq_method}, \eqref{eq:mse_sigma_method}, and the definition of relative MSE in \eqref{eq:rmse_def} we have
\begin{align}
 \text{Relative MSE of }\sigmaest{method}^2 \text{ wrt }\var\left[\Test{method}\right] &= \mathcal{O}\left(N^{-1}\right)\,,~~\forall \text{method}\in\{\text{basic, centered, cv}\}\,,\\
 \text{Relative MSE of }\sigmaest{method} \text{ wrt }\sqrt{\var\left[\Test{method}\right]} &= \mathcal{O}\left(N^{-1}\right)\,,~~\forall \text{method}\in\{\text{basic, centered, cv}\}\,.
\end{align}

\section{Parameters of the Noise Model} \label{appendix:noise_model_params}

\begin{table}[ht]
 \centering
 \caption{Values of $-\ln(1-2\,\epsilon_{\text{\sc t},i})/2$, from Ref.~\cite{vandenBerg2023}, for the 4-qubit noise model for different local Pauli interactions $P_i$ and different layer types $\mathrm{T}\in\{1,2\}$. In each row, the ``Pauli String'' entry denotes how the corresponding local Pauli interaction $P_i$ acts on the qubits of the circuit, in order from qubit $q_1$ to qubit $q_4$. The ``Layer Type 1'' and ``Layer Type 2'' entries are the corresponding values of $-\ln(1-2\,\epsilon_{\text{\sc t},i})/2$ for $\mathrm{T}=1$ and $\mathrm{T}=2$, respectively.}
 \label{tab:4_qubits_noise_model_params}
 \footnotesize
 \begin{tabular}{| c c c | c c c |}
  \hline
  \textbf{Pauli String} & \textbf{Layer Type 1} & \textbf{Layer Type 2} & \textbf{Pauli String} & \textbf{Layer Type 1} & \textbf{Layer Type 2} \\
  \hline
  \rowcolor{black!10} \texttt{XIII} & 0.000669974414 & 0.000659860950 & \texttt{ZZII} & 0.000870530575 & 0.000000000000 \\
  \texttt{YIII} & 0.000596776037 & 0.000752241008 & \texttt{IXXI} & 0.000000000000 & 0.001227428083 \\
  \rowcolor{black!10} \texttt{ZIII} & 0.001627948647 & 0.000392342844 & \texttt{IXYI} & 0.000038574173 & 0.000243853430 \\
  \texttt{IXII} & 0.000774443205 & 0.000680371432 & \texttt{IXZI} & 0.000000000000 & 0.000109003362 \\
  \rowcolor{black!10} \texttt{IYII} & 0.000340566246 & 0.000930175204 & \texttt{IYXI} & 0.000020093400 & 0.001419707915 \\
  \texttt{IZII} & 0.000993663526 & 0.002420048873 & \texttt{IYYI} & 0.000053643887 & 0.000103959855 \\
  \rowcolor{black!10} \texttt{IIIX} & 0.000702136549 & 0.000672692617 & \texttt{IYZI} & 0.000012700720 & 0.000242172351 \\
  \texttt{IIIY} & 0.000602160288 & 0.001066022658 & \texttt{IZXI} & 0.000000000000 & 0.001924003893 \\
  \rowcolor{black!10} \texttt{IIIZ} & 0.000478242572 & 0.000219538914 & \texttt{IZYI} & 0.000004671438 & 0.000408402037 \\
  \texttt{IIXI} & 0.000651171826 & 0.001018792408 & \texttt{IZZI} & 0.000023699323 & 0.000831125651 \\
  \rowcolor{black!10} \texttt{IIYI} & 0.000204387730 & 0.000686086310 & \texttt{IIXX} & 0.000620397324 & 0.000000000000 \\
  \texttt{IIZI} & 0.000466925204 & 0.000422472781 & \texttt{IIYX} & 0.000030479719 & 0.000004344610 \\
  \rowcolor{black!10} \texttt{XXII} & 0.000633952178 & 0.000000000000 & \texttt{IIZX} & 0.000002142946 & 0.000043848925 \\
  \texttt{XYII} & 0.000000000000 & 0.000048372336 & \texttt{IIXY} & 0.000662476196 & 0.000011875130 \\
  \rowcolor{black!10} \texttt{XZII} & 0.000010257834 & 0.000049852946 & \texttt{IIYY} & 0.000005369033 & 0.000011953629 \\
  \texttt{YXII} & 0.000631370212 & 0.000317148337 & \texttt{IIZY} & 0.000018339999 & 0.000000000000 \\
  \rowcolor{black!10} \texttt{YYII} & 0.000006120382 & 0.000256771868 & \texttt{IIXZ} & 0.000092828806 & 0.000069207293 \\
  \texttt{YZII} & 0.000000000000 & 0.000069369248 & \texttt{IIYZ} & 0.000379231035 & 0.000061892636 \\
  \rowcolor{black!10} \texttt{ZXII} & 0.000020565733 & 0.000182306026 & \texttt{IIZZ} & 0.000446400089 & 0.000020360472 \\
  \texttt{ZYII} & 0.000612440964 & 0.000258825995 & & & \\
  \hline
 \end{tabular}
\end{table}

\begin{table}[ht]
 \centering
 \caption{Same as \tref{tab:4_qubits_noise_model_params}, but for the noise model used for the 10-qubit circuits.}
 \label{tab:10_qubits_noise_model_params}
 \footnotesize
 \resizebox{.85\textwidth}{!}{\begin{tabular}{| c c c | c c c |}
  \hline
  \textbf{Pauli String} & \textbf{Layer Type 1} & \textbf{Layer Type 2} & \textbf{Pauli String} & \textbf{Layer Type 1} & \textbf{Layer Type 2}\\
  \hline
  \rowcolor{black!10} \texttt{XIIIIIIIII} & 0.001192209288 & 0.001574096427 & \texttt{IIZZIIIIII} & 0.000061362164 & 0.000041643557 \\
  \texttt{YIIIIIIIII} & 0.000749369410 & 0.003297398650 & \texttt{IIIXXIIIII} & 0.000123281363 & 0.001772043115 \\
  \rowcolor{black!10} \texttt{ZIIIIIIIII} & 0.000591214457 & 0.000439373589 & \texttt{IIIXYIIIII} & 0.000937205886 & 0.001211585333 \\
  \texttt{IXIIIIIIII} & 0.001098771827 & 0.001099386481 & \texttt{IIIXZIIIII} & 0.000000000000 & 0.000426801823 \\
  \rowcolor{black!10} \texttt{IYIIIIIIII} & 0.001096640031 & 0.001307873798 & \texttt{IIIYXIIIII} & 0.000122791833 & 0.000000000000 \\
  \texttt{IZIIIIIIII} & 0.000396081356 & 0.002244376821 & \texttt{IIIYYIIIII} & 0.000000000000 & 0.000087273160 \\
  \rowcolor{black!10} \texttt{IIXIIIIIII} & 0.000741038474 & 0.001257650908 & \texttt{IIIYZIIIII} & 0.000000000000 & 0.000533974096 \\
  \texttt{IIYIIIIIII} & 0.000016519117 & 0.001220585681 & \texttt{IIIZXIIIII} & 0.000000000000 & 0.000000000000 \\
  \rowcolor{black!10} \texttt{IIZIIIIIII} & 0.001088379181 & 0.000725489302 & \texttt{IIIZYIIIII} & 0.000000000000 & 0.000000000000 \\
  \texttt{IIIXIIIIII} & 0.000244355726 & 0.001270798557 & \texttt{IIIZZIIIII} & 0.000000000000 & 0.000720377032 \\
  \rowcolor{black!10} \texttt{IIIYIIIIII} & 0.000576171145 & 0.000195510139 & \texttt{IIIIXXIIII} & 0.001934621533 & 0.000000000000 \\
  \texttt{IIIZIIIIII} & 0.001427589586 & 0.000898078793 & \texttt{IIIIXYIIII} & 0.001629487149 & 0.000242928059 \\
  \rowcolor{black!10} \texttt{IIIIXIIIII} & 0.004558124788 & 0.001020041782 & \texttt{IIIIXZIIII} & 0.003052175387 & 0.000074305505 \\
  \texttt{IIIIYIIIII} & 0.000000000000 & 0.000187456486 & \texttt{IIIIYXIIII} & 0.000000000000 & 0.000000000000 \\
  \rowcolor{black!10} \texttt{IIIIZIIIII} & 0.002055441881 & 0.001195809556 & \texttt{IIIIYYIIII} & 0.000000000000 & 0.000000000000 \\
  \texttt{IIIIIXIIII} & 0.000000000000 & 0.001557243291 & \texttt{IIIIYZIIII} & 0.000138128144 & 0.000471576868 \\
  \rowcolor{black!10} \texttt{IIIIIYIIII} & 0.000000000000 & 0.001102903309 & \texttt{IIIIZXIIII} & 0.000484195825 & 0.000716452361 \\
  \texttt{IIIIIZIIII} & 0.000000000000 & 0.001335759609 & \texttt{IIIIZYIIII} & 0.000000000000 & 0.000803853053 \\
  \rowcolor{black!10} \texttt{IIIIIIXIII} & 0.001279398865 & 0.002663727502 & \texttt{IIIIZZIIII} & 0.000795053889 & 0.001907149099 \\
  \texttt{IIIIIIYIII} & 0.000513044213 & 0.000779190131 & \texttt{IIIIIXXIII} & 0.000015211361 & 0.002326031848 \\
  \rowcolor{black!10} \texttt{IIIIIIZIII} & 0.000804474619 & 0.001528538727 & \texttt{IIIIIXYIII} & 0.000011118094 & 0.000043861514 \\
  \texttt{IIIIIIIXII} & 0.000458229905 & 0.002072231505 & \texttt{IIIIIXZIII} & 0.000036471446 & 0.000205902312 \\
  \rowcolor{black!10} \texttt{IIIIIIIYII} & 0.000000000000 & 0.001722421273 & \texttt{IIIIIYXIII} & 0.000075912447 & 0.002657858584 \\
  \texttt{IIIIIIIZII} & 0.000383806139 & 0.000703703968 & \texttt{IIIIIYYIII} & 0.000000000000 & 0.000000000000 \\
  \rowcolor{black!10} \texttt{IIIIIIIIXI} & 0.000467617044 & 0.001041426266 & \texttt{IIIIIYZIII} & 0.000210159084 & 0.000203542224 \\
  \texttt{IIIIIIIIYI} & 0.000920457389 & 0.003963801984 & \texttt{IIIIIZXIII} & 0.000000000000 & 0.001477898996 \\
  \rowcolor{black!10} \texttt{IIIIIIIIZI} & 0.000168519683 & 0.000000000000 & \texttt{IIIIIZYIII} & 0.000270907748 & 0.001185063494 \\
  \texttt{IIIIIIIIIX} & 0.001194147229 & 0.001434615371 & \texttt{IIIIIZZIII} & 0.000000000000 & 0.001046785030 \\
  \rowcolor{black!10} \texttt{IIIIIIIIIY} & 0.000892751762 & 0.002631458309 & \texttt{IIIIIIXXII} & 0.001560365223 & 0.000017533842 \\
  \texttt{IIIIIIIIIZ} & 0.000547924822 & 0.000389343372 & \texttt{IIIIIIXYII} & 0.000000000000 & 0.000000000000 \\
  \rowcolor{black!10} \texttt{XXIIIIIIII} & 0.001007541252 & 0.000118112821 & \texttt{IIIIIIXZII} & 0.000268129809 & 0.000300757657 \\
  \texttt{XYIIIIIIII} & 0.000000000000 & 0.000043935951 & \texttt{IIIIIIYXII} & 0.000000000000 & 0.000000000000 \\
  \rowcolor{black!10} \texttt{XZIIIIIIII} & 0.000201257642 & 0.000125018218 & \texttt{IIIIIIYYII} & 0.000052918402 & 0.000000000000 \\
  \texttt{YXIIIIIIII} & 0.000629181408 & 0.000010172716 & \texttt{IIIIIIYZII} & 0.001530563576 & 0.000003891001 \\
  \rowcolor{black!10} \texttt{YYIIIIIIII} & 0.000205566380 & 0.000156020690 & \texttt{IIIIIIZXII} & 0.000058398263 & 0.000000000000 \\
  \texttt{YZIIIIIIII} & 0.000089129682 & 0.000063652221 & \texttt{IIIIIIZYII} & 0.000092291145 & 0.000000000000 \\
  \rowcolor{black!10} \texttt{ZXIIIIIIII} & 0.000431830479 & 0.000100350982 & \texttt{IIIIIIZZII} & 0.000553265191 & 0.000221411717 \\
  \texttt{ZYIIIIIIII} & 0.001738900933 & 0.000058926515 & \texttt{IIIIIIIXXI} & 0.000000000000 & 0.000765097294 \\
  \rowcolor{black!10} \texttt{ZZIIIIIIII} & 0.001289067246 & 0.000124157533 & \texttt{IIIIIIIXYI} & 0.000000000000 & 0.000000000000 \\
  \texttt{IXXIIIIIII} & 0.000099020020 & 0.000700906874 & \texttt{IIIIIIIXZI} & 0.000133939777 & 0.000356784735 \\
  \rowcolor{black!10} \texttt{IXYIIIIIII} & 0.000000000000 & 0.000275730285 & \texttt{IIIIIIIYXI} & 0.000000000000 & 0.001808081964 \\
  \texttt{IXZIIIIIII} & 0.000229637029 & 0.000159724387 & \texttt{IIIIIIIYYI} & 0.000000000000 & 0.000000000000 \\
  \rowcolor{black!10} \texttt{IYXIIIIIII} & 0.000132599092 & 0.002349787799 & \texttt{IIIIIIIYZI} & 0.000000000000 & 0.000000000000 \\
  \texttt{IYYIIIIIII} & 0.000000000000 & 0.000000000000 & \texttt{IIIIIIIZXI} & 0.000353891078 & 0.000152048958 \\
  \rowcolor{black!10} \texttt{IYZIIIIIII} & 0.000271474031 & 0.000000000000 & \texttt{IIIIIIIZYI} & 0.000028668672 & 0.002266435455 \\
  \texttt{IZXIIIIIII} & 0.000644700037 & 0.000268975988 & \texttt{IIIIIIIZZI} & 0.000121441485 & 0.001657914194 \\
  \rowcolor{black!10} \texttt{IZYIIIIIII} & 0.000000000000 & 0.000952029156 & \texttt{IIIIIIIIXX} & 0.001012587051 & 0.000000000000 \\
  \texttt{IZZIIIIIII} & 0.000393355287 & 0.000996173733 & \texttt{IIIIIIIIXY} & 0.001249378880 & 0.000179567753 \\
  \rowcolor{black!10} \texttt{IIXXIIIIII} & 0.001510954480 & 0.000000000000 & \texttt{IIIIIIIIXZ} & 0.000144233331 & 0.000000000000 \\
  \texttt{IIXYIIIIII} & 0.000000000000 & 0.000000000000 & \texttt{IIIIIIIIYX} & 0.000097945998 & 0.000065833132 \\
  \rowcolor{black!10} \texttt{IIXZIIIIII} & 0.000136586128 & 0.000000000000 & \texttt{IIIIIIIIYY} & 0.000095967505 & 0.000000000000 \\
  \texttt{IIYXIIIIII} & 0.000825053403 & 0.000049287415 & \texttt{IIIIIIIIYZ} & 0.000817642849 & 0.000139637787 \\
  \rowcolor{black!10} \texttt{IIYYIIIIII} & 0.000000000000 & 0.000064329889 & \texttt{IIIIIIIIZX} & 0.000032998376 & 0.000000000000 \\
  \texttt{IIYZIIIIII} & 0.000301372131 & 0.000000000000 & \texttt{IIIIIIIIZY} & 0.000000000000 & 0.000000000000 \\
  \rowcolor{black!10} \texttt{IIZXIIIIII} & 0.000768728717 & 0.000000000000 & \texttt{IIIIIIIIZZ} & 0.000598775101 & 0.000120436643 \\
  \texttt{IIZYIIIIII} & 0.000167192701 & 0.000000000000 & & & \\
  \hline
 \end{tabular}}
\end{table}

\end{document}